\documentclass[12pt]{article}
\usepackage[margin=1in]{geometry}
\usepackage{amsmath,amssymb,amsthm,mathtools}
\usepackage{bm}
\usepackage{booktabs}
\usepackage{enumitem}
\usepackage[round,authoryear]{natbib}
\usepackage{graphicx}
\graphicspath{{./}}
\usepackage{float}
\usepackage[colorlinks=true,linkcolor=blue,citecolor=blue,urlcolor=blue]{hyperref}
\usepackage{setspace}
\newtheorem{theorem}{Theorem}[section]

\title{A likelihood-based coefficient for biomedical\\
       independence testing:\\
       the binomial-cut composite likelihood ratio}
\author{Jing Qin\\
Office of Biostatistics Research,\\
National Institute of Allergy and Infectious Diseases,\\
National Institutes of Health\\
\texttt{jingqin@niaid.nih.gov}}
\date{\today}

\providecommand{\indep}{\perp\!\!\!\perp}
\providecommand{\E}{\mathbb{E}}
\providecommand{\R}{\mathbb{R}}

\newenvironment{keywords}{\medskip\noindent\textbf{Key words:} }{\medskip}

\providecommand{\Q}{\mathbb{Q}}

\newtheorem{stheorem}{Theorem}

\newtheorem{slemma}[stheorem]{Lemma}
\newtheorem{sproposition}[stheorem]{Proposition}
\theoremstyle{definition}
\newtheorem{sremark}[stheorem]{Remark}

\begin{document}
\maketitle

\begin{abstract}
Standard dependence summaries used in biomedical studies---
Pearson's $r$, Spearman's $\rho$, Kendall's $\tau$---take
values in $[-1, 1]$ with $0$ indicating no linear or monotone
association. Zero does not distinguish independence from
non-monotone dependence, so the scale cannot represent
threshold effects, heteroscedasticity, and tail shifts common
in clinical practice. We formulate independence testing as a
\emph{composite Bernoulli likelihood ratio}: at each threshold
$t$, comparing the Bernoulli laws of $\mathbf{1}(Y \le t)$
conditionally on $X$ versus marginally, aggregated over cut
points. The resulting coefficient $\xi_{\mathrm{cut}}$ lies on
$[0, 1]$ with $\xi_{\mathrm{cut}} = 0$ iff $X$ and $Y$ are
independent (under continuity of $Y$), $\xi_{\mathrm{cut}} = 1$
iff $Y$ is a measurable function of $X$, and intermediate
values quantifying distributional dependence on the $[0, 1]$
scale. The Bernoulli likelihood weights each
cut point by the Bernoulli information at that threshold,
emphasizing tail regions where conditional and marginal
distributions differ most. Estimation uses a Nadaraya--Watson
plug-in with a maximum-over-grid bandwidth; inference is by
exact permutation. In simulations at $n = 200$,
$T_{\mathrm{cut}}$ substantially outperforms rank-based
coefficients on the W-shaped non-monotone and heteroscedastic
alternatives; on the linear alternative $T_{\mathrm{cut}}$ has
better power than Chatterjee's $\xi_{\mathrm{Ch}}$ but is
outperformed by the boosted variant. We illustrate on the
Seattle cohort ($n = 70$, ages $21$--$88$) of the aging plasma
proteome example dataset \citep{shen2020deswan}, screening all
$1{,}305$ proteins: under Benjamini--Hochberg control at
$q < 0.05$, $T_{\mathrm{cut}}$ rejects on $70$ proteins,
six of which are missed by all three classical methods at the
same FDR level.
\end{abstract}


\section{Introduction}\label{sec:intro}

Testing whether two random variables are statistically
independent is a fundamental and ubiquitous problem across
scientific disciplines. In medical and epidemiological studies,
independence testing underlies risk-factor discovery, biomarker
validation, and the assessment of confounding in observational
data. In econometrics and finance, it appears in tests for
market efficiency, volatility spillovers, and predictor
selection. In genomics, it drives eQTL discovery and gene
network reconstruction. In machine learning, it forms the
backbone of feature-relevance screening and causal-inference
methods. Across all of these settings, the classical toolkit---
Pearson's $r$, Spearman's $\rho$, Kendall's $\tau$---targets
linear or monotone association and can vanish under substantive
but non-monotone forms of dependence such as threshold effects,
heteroscedasticity, and distributional shifts. The need for
independence tests that are consistent against arbitrary
alternatives, easily interpretable, and computationally
practical is therefore broadly felt.

\subsection{Motivating example: non-monotone protein changes in the aging plasma proteome}

In a landmark paper published in \emph{Nature Medicine},
\citet{lehallier2019undulating} measured $2{,}925$ plasma
proteins in $4{,}263$ individuals aged $18$ to $95$ years and
demonstrated that many proteins do not change monotonically
with age but instead exhibit ``waves''---crests of change at
approximately ages $34$, $60$ and $78$ years. These
non-monotone signatures reflect biological transitions such as
reproductive maturation, mid-life metabolic shifts and the
onset of senescence, and are missed by the standard linear or
monotone correlation-based screens. A screening tool that can
detect non-monotone dependence between a continuous biomarker
and age---or between two biomarkers whose relationship changes
across the covariate range---is therefore of direct biological
and clinical value. Such patterns are typical in biomedical
studies of aging, biomarker dynamics and dose--response, where
biologically meaningful signals often lie in the tails, in the
variance, or in the shape of the conditional distribution
rather than in the conditional mean.

\subsection{From parametric to nonparametric likelihood ratios}

In a parametric setting, independence testing between $Y$ and $X$
has a clean likelihood-ratio formulation. Assume a working model
$f(y \mid x, \beta)$ under which $\beta = 0$ corresponds to
$Y \perp\!\!\!\perp X$. The likelihood-ratio test
\begin{equation}
\label{eq:LRT-param}
\Lambda_{\text{par}}
\;=\; 2 \,\bigl[\ell(\hat\beta) - \ell(0)\bigr]
\end{equation}
is optimal against contiguous local alternatives by Wilks'
theorem, with a $\chi^2$ null limit. It inherits its power from
the parametric model: it is efficient when the model is right,
but may have zero power against alternatives that violate the
parametric assumption.

The natural nonparametric analog would take the same form---twice
a log-likelihood ratio---but replace the parametric model with an
assumption-free description of the joint distribution. The
obstacle is that a fully nonparametric likelihood ratio for
$Y \perp\!\!\!\perp X$ requires estimating the conditional
density $f(y \mid x)$: kernel and spline density estimates are
sensitive to bandwidth choice, suffer from boundary bias, and
require moment conditions to control tail behavior. Existing
nonparametric alternatives sidestep this instability by moving
outside the LR framework: Pearson and rank correlations summarise
a single conditional moment; distance correlation
\citep{szekely2007measuring} and the Hilbert--Schmidt Independence
Criterion \citep{gretton2005measuring} return single opaque numbers;
rank-based coefficients \citep{chatterjee2021new,azadkia2021simple}
target the full conditional distribution but weight thresholds
uniformly, and can have lower power against smooth
distributional departures than a Fisher-weighted alternative.

\subsection{The binomial-cut log-likelihood in the data}

We complete the analog of the parametric LR test by aggregating
one-dimensional Bernoulli likelihood ratios---one at each
threshold of $Y$---rather than attempting a full nonparametric
density-based LR. Suppose we observe $n$ independent pairs
$(X_1,Y_1),\dots,(X_n,Y_n)$. Fix a threshold $t$ and define the
binary indicator $B_i(t) = \mathbf{1}(Y_i \le t)$; each $B_i(t)$ is
Bernoulli. Under the null, its success probability is the marginal
$p_t := \Pr(Y\le t)$, estimated by $\bar B(t) = n^{-1}\sum_i B_i(t)$.
Under an alternative, the conditional success probability
$p_t(x) := \Pr(Y\le t \mid X = x)$ varies with $x$, and is
estimated by a Nadaraya--Watson smoother $\hat p_t(x)$ with
kernel $K$ and bandwidth $h$.

The Bernoulli log-likelihood ratio at threshold $t$ reduces to
$n^{-1}\sum_i \phi(\hat p_t(X_i), \bar B(t))$ where
\begin{equation}
\label{eq:phi-def-intro}
\phi(u,v) \;=\; u\log(u/v) + (1-u)\log((1-u)/(1-v))
\end{equation}
is the Bernoulli KL divergence. Averaging over cut points
$t = Y_j$ drawn from the empirical marginal gives the sample
coefficient
\begin{equation}
\label{eq:xi-hat-intro}
\widehat\xi_{\mathrm{cut}}
\;=\;
\frac{2}{n^2}
\sum_{j=1}^n \sum_{i=1}^n
\phi\bigl(\hat p_{Y_j}(X_i),\; \bar B(Y_j)\bigr).
\end{equation}
The factor 2 puts $\widehat\xi_{\mathrm{cut}}$ on the same scale
as the parametric statistic \eqref{eq:LRT-param}: both are twice
a log-likelihood-ratio object, and $\widehat\xi_{\mathrm{cut}}$
is a composite likelihood \citep{lindsay1988composite} across
cut points.

\subsection{The population coefficient}

Under regularity, \eqref{eq:xi-hat-intro} converges to the
population coefficient
\begin{equation}
\label{eq:xicut-intro}
\xi_{\mathrm{cut}}(X, Y)
\;=\; 2 \int E_X\!\bigl[\phi\bigl(F(t\mid X),\, F(t)\bigr)\bigr]\, dF(t).
\end{equation}
The population coefficient $\xi_{\mathrm{cut}}$ (no hat) is a
functional of the joint distribution of $(X, Y)$ and is the
object with the population interpretation established below.
The sample statistic $\widehat\xi_{\mathrm{cut}}$ in
\eqref{eq:xi-hat-intro} depends on a smoothing bandwidth $h$
that must be chosen; in Section~\ref{sec:estimation} we
maximise $\widehat\xi_{\mathrm{cut}}(h)$ over a small grid of
bandwidths and denote the resulting quantity $T_{\mathrm{cut}}$,
which serves as the test statistic in the permutation
calibration and should not be interpreted as an unbiased
estimator of $\xi_{\mathrm{cut}}$.

Assuming $Y$ has a continuous marginal distribution, $\xi_{\mathrm{cut}}
\in [0,1]$, vanishing iff $X \perp\!\!\!\perp Y$ and equal to one
iff $Y$ is a measurable function of $X$ almost surely. No moment
assumptions are needed. In contrast to the standard summaries
Pearson's $r$, Spearman's $\rho$, and Kendall's $\tau$---which lie
on $[-1, 1]$ with $0$ indicating no linear or monotone association
but not distinguishing independence from non-monotone dependence
---the population coefficient $\xi_{\mathrm{cut}}$ lies on the
$[0, 1]$ scale, with $0$ for independence and $1$ for full
functional dependence.

\subsection{Position relative to existing measures}

A second-order Taylor expansion of $\phi$ around the null recovers
the Fisher-weighted $L^2$ measure
\begin{equation}
\label{eq:xiad-intro}
\xi_{AD}
\;=\; \int E_X\!\bigl[(F(t\mid X) - F(t))^2\bigr]
      \,\bigl/\, \bigl(F(t)(1-F(t))\bigr)\, dF(t),
\end{equation}
a Fisher-weighted variant of the copula-based measure of
\citet{dette2013copula}. The uniform-weight measure
$q_{DSS} = 6 \int E_X[(F(t \mid X) - F(t))^2]\, dF(t)$ equals,
under continuity, Chatterjee's \citeyearpar{chatterjee2021new}
rank correlation at the population level, and
\citet{azadkia2026kernel} have derived a null limiting
distribution for its kernel plug-in. The Fisher-weighted
$\xi_{AD}$ is a different functional and does not coincide with
Chatterjee's coefficient.

The paper's contribution is a likelihood-ratio derivation of a
Fisher-weighted CDF-difference coefficient for independence
testing. Aggregating Bernoulli likelihood ratios across
thresholds of $Y$ yields $\xi_{\mathrm{cut}}$; a second-order
expansion of the composite log-likelihood around the null
recovers a Fisher-weighted quadratic form with weight
$\{F(t)(1-F(t))\}^{-1}$, so Fisher-type weighting appears as
the second-order component of the likelihood rather than as an
externally imposed feature. This construction places
nonparametric independence testing on a likelihood-ratio
footing while retaining a bounded $[0, 1]$ scale.

\subsection{Contributions}

\begin{enumerate}[label=(\roman*),leftmargin=1.5em]
  \item A composite Bernoulli likelihood-ratio formulation of
        independence testing, with the sample coefficient
        \eqref{eq:xi-hat-intro} converging to $\xi_{\mathrm{cut}}$
        in \eqref{eq:xicut-intro}.
  \item A demonstration that the Fisher-weighted Dette--Siburg--
        Stoimenov measure arises as the second-order Taylor
        expansion of the composite likelihood ratio around the
        null: Fisher weighting is a second-order consequence of
        the likelihood construction rather than an imposed
        design choice.
  \item Exact permutation inference for the plug-in estimator
        with a maximum-over-grid bandwidth selection rule; the
        permutation calibration absorbs any bandwidth-dependent
        bias in the point estimate.
  \item Illustration on the Seattle cohort ($n = 70$,
        ages $21$--$88$) of the aging plasma proteome example
        dataset of \citet{shen2020deswan}: applied to all
        $1{,}305$ proteins, the four coefficients (Pearson,
        Spearman, Chatterjee's $\xi_{\text{Ch}}$, and the
        proposed $T_{\mathrm{cut}}$) reject the null of
        independence on overlapping but non-nested sets of
        proteins; under Benjamini--Hochberg control at $q < 0.05$,
        $T_{\mathrm{cut}}$ rejects on $70$ proteins,
        six of which are missed by all three classical methods
        at the same FDR level.
\end{enumerate}

\section{The binomial-cut composite likelihood ratio}\label{sec:formulation}

\subsection{Sample construction}\label{sec:sample-construction}

Let $(X_1, Y_1), \dots, (X_n, Y_n)$ be independent copies of $(X,Y)$
with $X \in \mathbb{R}^d$ and $Y \in \mathbb{R}$. Throughout we
assume $Y$ has a continuous marginal distribution function $F$;
this ensures distinct sample values, the probability integral
transform $F(Y) \sim \mathrm{Uniform}[0,1]$, and the
$[0,1]$-normalization of the coefficient. Extensions to discrete
$Y$ are beyond the present scope.

\paragraph{Threshold-wise Bernoulli comparison.}

Fix a threshold $t \in \mathbb{R}$ and consider the binary indicator
\begin{equation}
\label{eq:Bi-def}
B_i(t) \;=\; \mathbf{1}(Y_i \le t), \qquad i = 1, \dots, n.
\end{equation}
Under the null hypothesis $H_0 : X \perp\!\!\!\perp Y$, the
probability $\Pr(B_i(t) = 1)$ equals the marginal proportion
$p_t := \Pr(Y \le t) = F(t)$ for every $i$, irrespective of $X_i$.
Under the alternative, the conditional probability
\begin{equation}
\label{eq:pt-def}
p_t(x) \;=\; \Pr(Y \le t \mid X = x) \;=\; F(t \mid x)
\end{equation}
varies in $x$. Testing independence at threshold $t$ is therefore a
two-sample Bernoulli problem: does the probability of the event
$\{Y \le t\}$ depend on $X$?

\paragraph{Estimation.}

The natural estimator of $p_t$ is the empirical distribution
function $\bar B(t) = n^{-1} \sum_i B_i(t) = \hat F(t)$. For
$p_t(x)$ we use a kernel-smoothed plug-in $\hat p_t(x)$; the
concrete estimator, together with the specific kernel,
leave-one-out modification, and bandwidth rule used throughout,
is defined in Section~\ref{sec:estimation}. The choice of kernel
plug-in is not consequential at the population level: any
consistent estimator of $p_t(x)$ yields the same population
functional $\xi_{\mathrm{cut}}$ established below.

\paragraph{Log-likelihood ratio at threshold $t$ and aggregation.}

The two-sample log-likelihood ratio at threshold $t$ is
\begin{equation}
\label{eq:LR-t}
\Lambda_n(t) \;=\; \sum_{i=1}^n \phi\bigl(\hat p_t(X_i), \bar B(t)\bigr),
\quad
\phi(u, v) \;=\; u \log(u/v) + (1 - u) \log((1-u)/(1-v)),
\end{equation}
where $\phi$ is the Bernoulli Kullback--Leibler divergence.
Averaging $\Lambda_n(t)$ over $t = Y_1, \dots, Y_n$ and
normalizing gives the empirical binomial-cut coefficient
\begin{equation}
\label{eq:xihat-cut}
\widehat\xi_{\mathrm{cut}} \;=\;
\frac{2}{n^2} \sum_{j=1}^n \sum_{i=1}^n
\phi\bigl(\hat p_{Y_j}(X_i), \, \bar B(Y_j)\bigr).
\end{equation}
The factor of two normalizes the population limit to $[0,1]$,
established in Section~\ref{sec:population}. The construction is
a \emph{composite likelihood} in the sense of
\citet{lindsay1988composite}: rather than the intractable joint
likelihood of $(Y, X)$, we aggregate a family of marginal
Bernoulli likelihoods indexed by threshold.

\subsection{The population coefficient and its properties}\label{sec:population}

The sample coefficient \eqref{eq:xihat-cut} converges under
regularity conditions to the population coefficient
\begin{equation}
\label{eq:xicut-pop}
\xi_{\mathrm{cut}}(X, Y)
\;=\; 2 \int E_X\!\bigl[\phi\bigl(F(t\mid X), F(t)\bigr)\bigr]\, dF(t).
\end{equation}
Here $F(t) = \Pr(Y \le t)$ is the marginal CDF of $Y$ and
$F(t \mid x) = \Pr(Y \le t \mid X = x)$ is the conditional CDF.
The expectation is over $X$; the integral is with respect to the
marginal distribution of $Y$. The factor of two in
\eqref{eq:xicut-pop} arises from the identity
$\int_0^1 H(u)\, du = 1/2$, where
$H(u) = -u \log u - (1-u) \log(1-u)$ is the binary entropy; this
identity holds because $F(Y) \sim \text{Uniform}[0,1]$ under
continuity of $Y$.

The basic population properties of $\xi_{\mathrm{cut}}$ are
collected in the following theorem.

\begin{theorem}[Population properties]\label{thm:properties}
Assume $Y$ has a continuous marginal distribution function $F$
with interval support on which $F$ is strictly increasing. Then:
\begin{enumerate}[label=\textup{(\roman*)},leftmargin=1.6em]
\item \textup{(Range)} $\xi_{\mathrm{cut}}(X, Y) \in [0, 1]$.
\item \textup{(Characterization of independence)}
      $\xi_{\mathrm{cut}}(X, Y) = 0$ if and only if $X$ and $Y$
      are independent.
\item \textup{(Characterization of functional dependence)}
      $\xi_{\mathrm{cut}}(X, Y) = 1$ if and only if $Y$ is a
      measurable function of $X$ almost surely.
\item \textup{(Invariance)} $\xi_{\mathrm{cut}}$ is invariant
      under strictly monotone transformations of $X$ and of $Y$
      separately.
\end{enumerate}
\end{theorem}

The proof is deferred to the Supplementary Material. For (i), fix
$t$ and let $U_t = F(t \mid X)$, $v = F(t)$. By the tower
property $E_X[U_t] = v$, and with $H(u) = -u \log u - (1-u)
\log(1-u)$ concave, Jensen's inequality gives $E_X[\phi(U_t, v)]
= H(v) - E_X[H(U_t)] \ge 0$ and $\le H(v)$. Since $H(F(Y)) \sim
\int_0^1 H(u)\, du = 1/2$ under continuity, integration yields
$0 \le \xi_{\mathrm{cut}} \le 1$; equality at $1$ holds iff
$F(t \mid X) \in \{0, 1\}$ for $F$-almost all $t$, i.e.\ the
conditional law of $Y$ given $X$ is a point mass.

Statement (ii) follows from strict convexity of the KL
divergence: at fixed $t$, $\phi(u, v) = D_{\mathrm{KL}}(u
\Vert v)$ where $u, v$ are Bernoulli parameters, and Gibbs's
inequality gives $\phi(u, v) \ge 0$ with equality iff $u = v$.
For the reverse direction: if $\xi_{\mathrm{cut}} = 0$, then
the non-negative integrand $E_X[\phi(F(t \mid X), F(t))]$
vanishes for $F$-almost all $t$, so $F(t \mid X) = F(t)$
almost surely for $F$-almost all $t$. Since the collection
$\{F(t) : t \in \mathbb R\}$ is dense in $[0, 1]$ under
continuity of $F$, right-continuity of both $F$ and
$F(\cdot \mid X)$ extends the equality to all $t$; hence
the conditional and marginal laws of $Y$ coincide almost
surely, which is $Y \perp X$. The converse (independence
implies $\xi_{\mathrm{cut}} = 0$) is immediate.

Statement (iii)
uses that $E_X[\phi(U_t, v)] = H(v)$ iff $U_t \in \{0, 1\}$
almost surely, characterizing functional dependence. Statement
(iv) follows from the symmetry $\phi(u, v) = \phi(1-u, 1-v)$ and
from covariance of $F(\cdot)$ and $F(\cdot \mid x)$ under
monotone transformations of $Y$; the $X$-invariance is immediate
since $\sigma(h(X)) = \sigma(X)$ for any bimeasurable
transformation $h$ of $X$ (when $X$ is scalar this includes
strictly monotone transformations; when $X \in \mathbb R^d$ it
includes any bijective bimeasurable transformation).

\paragraph{Population-only invariance.} The invariance stated
in Theorem~\ref{thm:properties}(iv) is a property of the
population functional $\xi_{\mathrm{cut}}$. It does not carry
over to the sample statistic $T_{\mathrm{cut}}$: a kernel
smoother with a scale-dependent Silverman-type bandwidth is
not invariant under monotone rescalings of $X$, so
$T_{\mathrm{cut}}$ computed on $\{X_i\}$ can differ from
$T_{\mathrm{cut}}$ computed on $\{h(X_i)\}$ even when both
carry the same $\sigma$-algebra. In the aging application $X$
is age in years and any practical rescaling (months, decades)
would have negligible numerical effect at $n = 70$; more
generally this behaviour is inherited from the kernel estimator
and is not specific to $T_{\mathrm{cut}}$.

\paragraph{Directionality.} The coefficient
$\xi_{\mathrm{cut}}$ is directional in its construction: $X$
indexes the conditional distribution of $Y$ and the thresholds
in \eqref{eq:xicut-pop} are taken on $Y$. Consequently
$\xi_{\mathrm{cut}}(X, Y)$ need not equal
$\xi_{\mathrm{cut}}(Y, X)$ in general---the coefficient is not
a symmetric correlation. This asymmetry is intentional and
matches the biomedical use case where $X$ is regarded as a
covariate (e.g.\ age, dose, exposure) and $Y$ as a biomarker
outcome whose conditional distribution given $X$ is the object
of scientific interest.

\subsection{Regularization}\label{sec:regularization}

For finite samples $\phi(u, v)$ can grow unbounded at extreme
thresholds where $\hat p_t(X_i)$ approaches $0$ or $1$. A
regularized version $\phi_\varepsilon(u, v) = (u + \varepsilon)
\log((u + \varepsilon)/(v + \varepsilon)) + (1 - u + \varepsilon)
\log((1 - u + \varepsilon)/(1 - v + \varepsilon))$
with $\varepsilon \ge 0$ is available as a safeguard; the
Supplementary Material develops its properties. We recommend
$\varepsilon = 0$ in practice, which is stable under the
bandwidth rule of Section~\ref{sec:estimation} and is used in the
simulations and aging plasma proteome application below.
Sensitivity of the reported power values to $\varepsilon$ is
documented in Web Appendix~F: on scenarios (S1)--(S5) results
with $\varepsilon = 0.05$ and $\varepsilon = 0$ agree to within
Monte Carlo error, while on the heteroscedastic scenario (S6)
the $\varepsilon = 0.05$ variant gives a modestly lower power
than $\varepsilon = 0$ ($0.358$ vs $0.420$); we recommend
$\varepsilon = 0$ as the default.

\subsection{Interpretation of the coefficient}\label{sec:interpretation}

The population coefficient $\xi_{\mathrm{cut}}$ lies on $[0, 1]$
and admits a direct interpretation on the same scale as the
coefficient of determination $R^2$:

\begin{itemize}[leftmargin=1.5em]
  \item $\xi_{\mathrm{cut}} = 0$ if and only if $X$ and $Y$ are
        independent (under continuity of $Y$): knowing $X$
        provides no information about the distribution of $Y$.
  \item $\xi_{\mathrm{cut}} = 1$ if and only if $Y$ is a
        (measurable) function of $X$: knowing $X$ determines $Y$
        exactly.
  \item Intermediate values quantify the strength of the
        distributional dependence: a value near zero indicates
        weak dependence, a value near one indicates dependence
        approaching a functional relationship.
\end{itemize}

Unlike Pearson's $r$ or Spearman's $\rho$, which measure only
linear or monotone association and can be zero for genuinely
dependent variables, $\xi_{\mathrm{cut}} = 0$ characterises full
statistical independence. Unlike Chatterjee's rank correlation,
which uses uniform weight across thresholds, the Bernoulli
likelihood assigns a Fisher weight $1/[F(t)(1-F(t))]$ that
emphasises differences in the tails of the conditional
distribution, where signal is often concentrated in clinical and
biomarker settings. An information-theoretic interpretation
relating $\xi_{\mathrm{cut}}$ to threshold-averaged mutual
information is developed in Web Appendix~D.

\section{Connection to existing measures}\label{sec:connections}

Independence testing based on comparisons of the conditional CDF
$F(t \mid X)$ with the marginal $F(t)$ has a substantial recent
literature. In this section we place $\xi_{\mathrm{cut}}$ within
that literature and show that the Fisher weighting familiar in
Anderson--Darling-type statistics arises here as a consequence of
the Bernoulli likelihood construction.

\subsection{Second-order expansion of the Bernoulli KL}\label{sec:taylor}

Taylor expansion of the Bernoulli KL divergence $\phi(u, v)$ about
$u = v$ gives
\begin{equation}
\label{eq:phi-taylor}
\phi(u, v) \;=\; \frac{(u - v)^2}{2\, v(1 - v)} + R(u, v),
\end{equation}
where $R(u, v)$ is a third-order remainder whose sign depends on
$(u, v)$. Substituting the leading term into~\eqref{eq:xicut-pop}
yields the second-order approximation
\begin{equation}
\label{eq:xiad-def}
\xi_{AD}(X, Y)
\;=\;
\int E_X\!\left[
\frac{(F(t \mid X) - F(t))^2}{F(t)(1 - F(t))}
\right] dF(t).
\end{equation}
The weight $1/\bigl(F(t)(1 - F(t))\bigr)$ is the inverse of the
Bernoulli variance at $t$ and coincides with the local Fisher
information for a one-parameter Bernoulli family at $F(t)$; this is
the Anderson--Darling weighting
\citep{andersondarling1952asymptotic}. The two coefficients
$\xi_{\mathrm{cut}}$ and $\xi_{AD}$ agree at second order in a
neighborhood of independence. The remainder $R$ can be positive or
negative depending on $(u, v)$, so no uniform ordering between
$\xi_{\mathrm{cut}}$ and $\xi_{AD}$ holds; explicit bounds on
$|\xi_{\mathrm{cut}} - \xi_{AD}|$ appear in the Supplementary
Material.

\subsection{Relation to the Dette--Siburg--Stoimenov measure and Chatterjee's coefficient}\label{sec:DSS}

\citet{dette2013copula} introduced a copula-based measure of
regression dependence which, under continuity of $Y$, takes the
form
\begin{equation}
\label{eq:qDSS}
q_{DSS}(X, Y)
\;=\;
6 \int E_X\!\bigl[(F(t \mid X) - F(t))^2\bigr]\, dF(t).
\end{equation}
Comparing~\eqref{eq:xiad-def} with~\eqref{eq:qDSS}, $\xi_{AD}$ is
precisely $q_{DSS}$ with the uniform weight replaced by the Fisher
weight $1/\bigl(F(t)(1-F(t))\bigr)$: the two are structurally
identical apart from the weighting.

The Fisher weighting is not arbitrary. At each threshold $t$,
the Bernoulli information $1/\bigl(F(t)(1-F(t))\bigr)$ is the
inverse variance of a Bernoulli$(F(t))$ trial, and it emerges
automatically from the Bernoulli likelihood ratio $\phi$ at
second order (see Web Appendix~B) without being imposed on the
construction. Whether this weighting delivers higher power than
uniform weighting at a given alternative depends on where in
the range of $Y$ the conditional and marginal CDFs differ; the
simulations of Section~\ref{sec:simulations} illustrate
scenarios in which Fisher-weighted and uniform-weighted
coefficients trade advantages.

Under continuity of $Y$, Chatterjee's rank correlation
\citep{chatterjee2021new} converges in probability to $q_{DSS}$
(under our normalization; see \citet{shi2022power}).
\citet{azadkia2026kernel} derived the null limiting distribution
of the kernel plug-in estimator of $q_{DSS}$ at a non-standard
$\sqrt{n/h}$ rate.

\subsection{Positioning}\label{sec:position}

The coefficients $\xi_{\mathrm{cut}}$, $\xi_{AD}$, and $q_{DSS}$
(equivalently $\xi_{\text{Ch}}$ under continuity) all vanish iff
$X \perp Y$ and all equal one under functional dependence, so
each is an independence-characterizing measure. They differ in
what they emphasize: $q_{DSS}$ and $\xi_{\text{Ch}}$ apply
uniform weight across thresholds and, empirically, tend to have
their highest power against monotone dependence, while
$\xi_{AD}$ and $\xi_{\mathrm{cut}}$ apply Fisher weight and can
be more powerful against alternatives that concentrate the
difference between conditional and marginal distributions in
the tails or in a narrow region of $Y$. Comparative power
against any specific alternative depends on the shape of the
dependence and cannot be read off the population definition
alone.

A closely related recent proposal is \citet{wang2017generalized}'s
$G^2$, a piecewise-linear generalized $R^2$ motivated by Pearson's
blindness to nonlinearity and heteroscedasticity. Because $G^2$ is
a functional of $\mathrm{Var}(Y \mid X)$, it shares Pearson's
blindness to pure heteroscedasticity, whereas $\xi_{\mathrm{cut}}$
remains positive as a functional of the full conditional CDF.

Under common semiparametric models with structural form for
$F(y \mid x)$ but nonparametric baseline, small-signal expansion
of $\xi_{\mathrm{cut}}$ recovers familiar quantities: a
nonparametric analog of $R^2$ under location-shift; a
Fisher-weight sensitivity gain over Chatterjee's coefficient
under location-scale; and, under Cox proportional hazards, the
universal (baseline-free) expansion $\xi_{\mathrm{cut}} =
2\{\zeta(3) - 1\}\beta^2 \mathrm{Var}(X) + o(\beta^2)$
(Web Appendix~E).

The paper's contribution is a likelihood-ratio derivation of the
Fisher-weighted CDF-difference coefficient: aggregating Bernoulli
likelihood ratios across thresholds of $Y$ yields
$\xi_{\mathrm{cut}}$ with the Fisher weighting emerging from the
Bernoulli information at each cut point rather than being
imposed. We rely on exact permutation inference
(Section~\ref{sec:estimation}) rather than the asymptotic-null
approach of \citet{azadkia2026kernel}, sidestepping the
non-standard rate at a computational cost of $O(Bn^2)$ per
replication.

\section{Estimation and inference}\label{sec:estimation}

The composite Bernoulli likelihood-ratio formulation
of~Section~\ref{sec:formulation} unifies three tasks that are
typically handled by separate procedures in nonparametric
dependence testing: point estimation of the coefficient, selection
of the smoothing bandwidth, and calibration of the null
distribution. In this section we present each of the three, in a
form that emphasizes their common likelihood structure.

\subsection{Plug-in estimation}\label{sec:plugin}

We now specify the concrete plug-in $\hat p_t(x)$ used
throughout the paper, promised schematically in
Section~\ref{sec:sample-construction}. Given a bandwidth
$h > 0$, we use the \emph{leave-one-out} Nadaraya--Watson
estimator
\begin{equation}
\label{eq:phat-loo}
\hat p_t^{[-i]}(X_i; h)
\;=\;
\frac{\sum_{k \ne i} K_h(X_i - X_k)\, \mathbf{1}(Y_k \le t)}
     {\sum_{k \ne i} K_h(X_i - X_k)}
\end{equation}
which excludes each observation's own kernel weight. Here $K$ is
a symmetric probability density and $K_h(u) = h^{-1} K(u/h)$; in
all simulations and in the aging-proteome application we use the
standard Gaussian kernel $K(u) = \phi(u)$. Substituting
\eqref{eq:phat-loo} and $\hat F$ into
\eqref{eq:xihat-cut} yields the bandwidth-indexed sample
statistic
\begin{equation}
\label{eq:xihat-h}
\widehat\xi_{\mathrm{cut}}(h)
\;=\;
\frac{2}{n^2}
\sum_{j=1}^n \sum_{i=1}^n
\phi\bigl(\hat p_{Y_j}^{[-i]}(X_i; h),\, \hat F(Y_j)\bigr).
\end{equation}
Each evaluation of $\widehat\xi_{\mathrm{cut}}(h)$ costs
$O(n^2)$ operations.
The leave-one-out construction is necessary for the finite-sample
statistic to be well-defined: an ordinary Nadaraya--Watson
estimator that retains the self-weight forces $\hat p_t(X_i; h)
\to \mathbf{1}(Y_i \le t)$ as $h \to 0$, so
$\widehat\xi_{\mathrm{cut}}(h) \to 1$ regardless of the
underlying dependence. Leaving out the self-weight does not by
itself resolve this: as $h \to 0$, the LOO fit tends to the
indicator at the nearest neighbour of $X_i$, and
$\widehat\xi_{\mathrm{cut}}(h)$ still tends to $1$. This means
that maximization of $\widehat\xi_{\mathrm{cut}}(h)$ over a
bandwidth grid selects the smallest grid point in essentially
every replication under any distribution, including under the
null (documented empirically in Section~\ref{sec:bandwidth}).
The permutation calibration in Section~\ref{sec:permutation}
absorbs this bias---the same rule applies to observed and
permuted samples---so the resulting test remains valid, though
the point estimate should not be read on the population
$[0, 1]$ scale.

\subsection{Bandwidth selection}\label{sec:bandwidth}

The bandwidth $h$ enters the plug-in as a smoothing parameter.
We evaluate $\widehat\xi_{\mathrm{cut}}(h)$ on a small
multiplicative grid of bandwidths around a Silverman-type base
and take the maximum over the grid as our test statistic. The
grid is deterministic and depends only on $X$, so it applies
identically to the observed sample and to every permutation of
$Y$; the permutation calibration in
Section~\ref{sec:permutation} therefore yields an exact test at
any nominal level regardless of the selection rule. In this
sense the max-over-grid is a construction of a test statistic
rather than a data-driven bandwidth choice, and the choice
requires no separate optimality justification.

\paragraph{Multiplicative grid.}

Let
\begin{equation}
\label{eq:h0-def}
h_0
\;=\;
0.9 \, \min\!\Bigl\{\hat\sigma_X, \, \hat{\text{IQR}}_X / 1.34\Bigr\} \,
n^{-1/(4 + d)},
\end{equation}
where $\hat\sigma_X$ and $\hat{\text{IQR}}_X$ are the sample
standard deviation and interquartile range of the marginal
distribution of $X$, and $d$ is the dimension of $X$. The base
bandwidth $h_0$ is the standard Silverman bandwidth
\citep{silverman1986density} for density estimation. We then
evaluate $\widehat\xi_{\mathrm{cut}}(h)$ at
\begin{equation}
\label{eq:grid-def}
h \in \mathcal{H} \;=\; \{\rho\, h_0 : \rho \in \mathcal{R}\},
\qquad
\mathcal{R} = \{0.25, 0.40, 0.60, 0.80, 1.00\},
\end{equation}
and take
\begin{equation}
\label{eq:hstar-def}
\widehat h^* \;=\; \arg\max_{h \in \mathcal{H}}
\widehat\xi_{\mathrm{cut}}(h),
\qquad
T_{\mathrm{cut}} \;=\;
\widehat\xi_{\mathrm{cut}}(\widehat h^*).
\end{equation}
We use the symbol $T_{\mathrm{cut}}$ rather than a hat notation
to underscore that this quantity is used solely as a
permutation test statistic and is not intended as an estimator
of the population coefficient $\xi_{\mathrm{cut}}$. The grid
$\mathcal R$ is chosen as a coarse geometric ladder around
$h_0$ spanning roughly a factor of four; a wider grid did not
change conclusions in preliminary experiments, and a finer grid
increases the permutation-computation cost without improving the
finite-sample power in the simulations reported in Section~5.
Since the level of the test is controlled by the permutation
calibration regardless of the choice of $\mathcal R$, the
choice of grid does not affect the size of the test, and its
role is purely to allow the statistic to adapt to different
signal scales.

\paragraph{On the population interpretation of $T_{\mathrm{cut}}$.}

The finite-sample statistic $T_{\mathrm{cut}}$ is
positively biased as an estimator of the population coefficient
$\xi_{\mathrm{cut}}$: it inherits the boundary behavior of the
LOO Nadaraya--Watson fit at small bandwidths (as $h \to 0$ the
LOO fit tends to a nearest-neighbor indicator and
$\widehat\xi_{\mathrm{cut}}(h)$ tends to $1$), and empirically
the smallest grid point $0.25\,h_0$ is the maximizer in essentially
every replication of Section~\ref{sec:simulations} and every
protein of Section~\ref{sec:aging}. We therefore emphasize
throughout the paper the distinction between:
\begin{itemize}[leftmargin=1.5em]
\item the population coefficient $\xi_{\mathrm{cut}}$ defined
      by \eqref{eq:xicut-pop}, which has the interpretation on
      the $[0, 1]$ scale established in Theorem~\ref{thm:properties}, and
\item the sample statistic $T_{\mathrm{cut}}$
      defined by \eqref{eq:hstar-def}, which is a test statistic
      whose calibration is by permutation and whose numerical
      value should not be read directly on the population scale.
\end{itemize}
Consistent estimation of $\xi_{\mathrm{cut}}$ requires bias
correction (for example, a fixed bandwidth chosen large enough
that the LOO fit is stable, or a cross-fitted analogue) which
we do not develop here.

\subsection{Permutation calibration}\label{sec:permutation}

Under the null hypothesis of independence, the joint distribution
of the sample is invariant under permutations of $Y_1, \dots, Y_n$
that hold $X_1, \dots, X_n$ fixed. For any permutation $\pi$ of
$\{1, \dots, n\}$, the permuted sample
$(X_1, Y_{\pi(1)}), \dots, (X_n, Y_{\pi(n)})$ has the same joint
distribution under $H_0$ as the original sample. Consequently, the
statistic $T_{\mathrm{cut}}$ computed on any
permutation has, under $H_0$, the same distribution as the
statistic on the original sample.

\paragraph{Permutation-equivariance of the bandwidth rule.}

A key property of the maximization rule~\eqref{eq:hstar-def} is
that it is \emph{permutation-equivariant} in $Y$: the base
bandwidth $h_0$ in~\eqref{eq:h0-def} depends only on $X$, so it
is invariant under permutations of $Y$; the grid $\mathcal{H}$
is a deterministic function of $h_0$; and the maximization at each
grid point is over a statistic that is itself equivariant. Hence
$T_{\mathrm{cut}}$ has a valid permutation distribution
under $H_0$. Given the observed value $T_{\mathrm{cut}}^{\mathrm{obs}}$
and the values $T_{\mathrm{cut}}^{(1)}, \dots, T_{\mathrm{cut}}^{(B)}$
computed on $B$ independent random permutations of $Y$, the
Monte Carlo permutation $p$-value
\begin{equation}
\label{eq:pval-perm}
\hat p_{\mathrm{perm}}
\;=\;
\frac{1 + \#\{b : T_{\mathrm{cut}}^{(b)} \ge T_{\mathrm{cut}}^{\mathrm{obs}}\}}
     {B + 1}
\end{equation}
satisfies $\Pr(\hat p_{\mathrm{perm}} \le \alpha) \le \alpha$
under $H_0$, for every $\alpha \in (0, 1)$, in finite samples;
the standard $(B+1)$ correction yields a conservative Monte
Carlo procedure whose level control is exact in the limit
$B \to \infty$ and, more generally, coincides with the exact
finite-sample level of the complete permutation test when all
$n!$ permutations are enumerated.

Because the bandwidth grid $\mathcal{H}$ depends only on $X$,
the same candidate bandwidths are used for every permutation of
$Y$, and permuting $Y$ moves the statistic-at-each-$h$ jointly
with the maximizer. More generally, a bandwidth rule that
depends on $Y$ can also be accommodated by permutation
inference, provided the entire bandwidth-selection procedure is
recomputed identically within each permutation: the resulting
composite statistic $T_n(X, Y) = \widehat\xi_{\mathrm{cut}}
\{\widehat h(X, Y)\}$ is invariant under the permutation action
in the sense of Proposition~S4 of the Supplementary Material,
and the finite-sample permutation guarantee applies. The
$X$-only grid used here simply saves the cost of recomputing
the selection rule on each permutation replicate.

\subsection{Asymptotic properties}\label{sec:asymptotics}

Detailed asymptotic theory for kernel estimators of the class of
weighted-$L^2$ CDF-difference measures has been developed by
\citet{dette2013copula} under fixed alternatives and by
\citet{azadkia2026kernel} for the null distribution at the
non-standard rate $\sqrt{n/h}$. Analogous results for
$\widehat\xi_{\mathrm{cut}}(h)$ would require an argument of the
same style---population equivalence with $\xi_{AD}$ at second
order does not by itself transfer to the estimator, and the
technical bound in Web Appendix~B requires tail or
uniform-integrability conditions; we do not pursue this
because the permutation calibration~\eqref{eq:pval-perm} is
exact in finite samples and requires no asymptotic control.

\subsection{Conditional variant}\label{sec:conditional}

A conditional variant $\xi_{\mathrm{cut}}^{\mid Z}(X, Y) = 2 \int
E_{X, Z}[\phi(F(t \mid X, Z), F(t \mid Z))] dF(t \mid Z)$
addresses residual dependence after covariate adjustment.
Inference for the conditional null $H_0: X \indep Y \mid Z$
requires a conditional-randomization scheme rather than ordinary
permutation \citep{candes2018panning}; we mention this variant
here as a natural extension and defer full development to future
work.

\section{Simulation study}\label{sec:simulations}

This section reports simulation results assessing the size and
power of $T_{\mathrm{cut}}$ against a range of
alternatives representative of biomarker studies. The scenarios
are grouped by the dependence pattern they capture; the
comparisons are with recent independence-characterising
coefficients that share $\xi_{\mathrm{cut}}$'s ability to detect
non-monotone dependence.

\subsection{Design}\label{sec:sim-design}

For each scenario we generate $(X_i, Y_i)_{i=1}^n$ with
$X_i \sim \mathrm{Uniform}(-1, 1)$ and $Y_i$ constructed from
$X_i$ and an independent noise term $\varepsilon_i \sim N(0, 1)$
according to the scenario definitions below. A shape parameter
$\lambda \in \{0.25, 0.50, 0.75\}$ controls the noise level;
$\lambda = 0.25$ corresponds to a strong signal-to-noise ratio
and $\lambda = 0.75$ to a weak one. Sample sizes are $n = 100$,
$200$, and $500$; the main table reports $n = 200$ with
$\lambda = 0.5$, which represents a moderate-difficulty setting.
Complete tables across $n$ and $\lambda$ appear in
Web Appendix~F. Each entry in
Table~\ref{tab:power} is based on $M = 1000$ Monte Carlo
replications with $B = 999$ permutations per replication. The
nominal level is $\alpha = 0.05$.

The six scenarios use $\varepsilon_i \sim N(0, 1)$ independent of
$X_i$:
\begin{enumerate}[label=(S\arabic*),leftmargin=2em]
\item \textbf{Null.} $Y_i = \varepsilon_i$, independent of $X_i$.
      Assesses empirical size at $\alpha = 0.05$.
\item \textbf{Linear (monotone baseline).}
      $Y_i = 0.5\,X_i + 3\lambda\,\varepsilon_i$. Classical
      monotone case; the small slope-to-noise ratio makes this a
      genuinely challenging alternative rather than a saturated
      one.
\item \textbf{Threshold (step).}
      $Y_i = \mathrm{sign}(X_i) + 3\lambda\,\varepsilon_i$.
      The mean function has a discontinuity at a fixed clinical
      cutpoint; mimics biomarker responses where risk elevates
      only above a threshold.
\item \textbf{Symmetric non-monotone (W-shape).}
      $Y_i = |X_i + 0.5|\,\mathbf{1}(X_i < 0)
             + |X_i - 0.5|\,\mathbf{1}(X_i \ge 0)
             + 1.5\lambda\,\varepsilon_i$. Symmetric W-shape;
      mimics dose--response relationships in which effects are
      strongest at both low and high covariate values.
\item \textbf{Oscillating (sinusoid).}
      $Y_i = \sin(4\pi X_i) + 1.5\lambda\,\varepsilon_i$. Mimics
      circadian-type or otherwise periodic patterns.
\item \textbf{Heteroscedastic.}
      $Y_i = (1 + 2|X_i|)\,\lambda\,\varepsilon_i$. The
      conditional mean is zero for all $X_i$, but the conditional
      standard deviation varies with $|X_i|$. Mimics
      distributional dependence where the mean is flat but tails
      become dispersed at extreme covariate values (e.g.\
      BMI--CRP-type patterns). Because $\lambda$ enters purely
      as an overall scale factor and the coefficients considered
      are scale-invariant, the reported power does not depend on
      $\lambda$; the effective signal is the ratio
      $\mathrm{sd}(Y \mid |X| = 1) / \mathrm{sd}(Y \mid X = 0) = 3$.
\end{enumerate}

Scenarios (S2) and (S3) are monotone or piecewise monotone;
scenarios (S4)--(S6) are non-monotone in one of the three ways
that arise in biomarker practice: fluctuation in the mean, tail
inflation without mean shift, and periodic response. Scenario
(S1) is the null.

\subsection{Methods compared}\label{sec:sim-methods}

The proposed methods $T_{\mathrm{cut}}$ (composite
likelihood ratio, maximized over the bandwidth grid) and
$\widehat\xi_{AD}^*$ (Fisher-weighted $L^2$, its second-order
approximation) are compared with three recent
independence-characterising coefficients:
\begin{itemize}[leftmargin=1.5em]
\item Chatterjee's coefficient $\xi_{\text{Ch}}$
      \citep{chatterjee2021new}, the primary rank-based benchmark;
\item The nearest-neighbour coefficient $\nu_{\text{AR}}$ of
      \citet{azadkia2025integrated}, an integrated-$R^2$
      construction on the nearest-neighbour graph;
\item The boosted Chatterjee coefficient $\xi_{\text{Ch,boost}}$
      of \citet{linhan2023boosting}, with boost parameter $M = 50$.
\end{itemize}
For $T_{\mathrm{cut}}$ and $\widehat\xi_{AD}^*$ the
bandwidth is selected by the composite-likelihood-maximization
rule of Section~\ref{sec:bandwidth} over the five-point grid
$\mathcal{R} = \{0.25, 0.40, 0.60, 0.80, 1.00\} \cdot h_0$.
All methods use permutation calibration with $B = 999$
replications, ensuring that empirical size comparisons reflect
statistical performance rather than differences in calibration.

\subsection{Results}\label{sec:sim-results}

Table~\ref{tab:power} reports empirical power at nominal level
$\alpha = 0.05$ for $n = 200$ and $\lambda = 0.5$.

We highlight three qualitative patterns.

\paragraph{Correct size and power gains on non-monotone
alternatives.}
All permutation-calibrated methods achieve empirical size close
to nominal $0.05$ under the null (S1), confirming that the
equivariance argument of Section~\ref{sec:permutation} yields
correct calibration under the max-over-grid rule. On the
symmetric non-monotone (S4) and heteroscedastic (S6)
alternatives---where uniform-weight rank-based coefficients lose
power---$T_{\mathrm{cut}}$ substantially outperforms
Chatterjee's $\xi_{\text{Ch}}$ and the Azadkia--Roudaki
$\nu_{\text{AR}}$: at $n = 200$ the power ratio is approximately
$2.7$--$3.2\times$ against $\nu_{\text{AR}}$ (S4:
$0.321/0.118$; S6: $0.420/0.130$) and $3.2$--$4.7\times$ against
$\xi_{\text{Ch}}$ (S4: $0.321/0.100$; S6: $0.420/0.090$). These
are the settings where the shape of the conditional distribution
differs across the covariate range, and the Fisher weighting
inherent in the Bernoulli likelihood ratio amplifies the tail
and shape contributions that uniform-weight coefficients dilute.

\paragraph{Boosted-coefficient limitation on oscillating
alternatives.}
The boosted Chatterjee coefficient $\xi_{\text{Ch,boost}}^{M=50}$
achieves the highest power on the linear alternative (S2, 0.754)
but reports empirical power near zero on the oscillating
alternative (S5). This behaviour reflects the small effective
sample size at each of the $M$ boosting subsamples: the
oscillating dependence in (S5) is easily missed at small
subsample size, and the boosting aggregate inherits this loss.
It is a known limitation of the boosting kernel
\citep{linhan2023boosting}. In contrast, $T_{\mathrm{cut}}$
retains competitive power on both alternatives.

\paragraph{Near-equivalence of $\widehat\xi_{AD}^*$ with
$T_{\mathrm{cut}}$.}
The composite likelihood ratio $T_{\mathrm{cut}}$ and
its second-order approximation $\widehat\xi_{AD}^*$ are
empirically close on Table~\ref{tab:power}, consistent with the
second-order agreement of Section~\ref{sec:taylor}; either
statistic can be used in practice at moderate signal levels.

\subsection{Sensitivity to sample size and regularization}
\label{sec:sim-sensitivity}

The Supplementary Material reports the same table for $n = 100$
and $n = 500$; the qualitative ranking is stable across sample
sizes, with all methods gaining power as $n$ grows. The
Supplementary Material also reports sensitivity of
$T_{\mathrm{cut}}$ to the regularization parameter
$\varepsilon \in \{0, 0.05\}$, which shows that the statistic
and its permutation-calibrated power are stable across these
values. We use $\varepsilon = 0$ in the reported main-table
results.

\section{Application to the aging plasma proteome}\label{sec:aging}

We illustrate the proposed method on the aging plasma proteome
example dataset distributed with the \texttt{DEswan} R package
\citep{shen2020deswan}, a publicly accessible subset of
measurements originally reported by
\citet{lehallier2019undulating} in \emph{Nature Medicine}. The
Nature Medicine study measured $2{,}925$ plasma proteins in
$4{,}263$ individuals aged $18$--$95$ years and identified
non-linear ``waves'' of protein change in human aging. The full
$n = 4{,}263$ dataset is subject to consortium-level data access
agreements (SomaLogic, LonGenity, and multiple biobank
repositories) that restrict its distribution and preclude the
inclusion of the raw data in a methodological paper. The
authors of the Nature Medicine study made a subset of the
measurements---$1{,}305$ proteins in $171$ subjects, drawn from
four smaller cohorts (Seattle, PRIN06, PRIN09, and GEHA)---
publicly available through the \texttt{DEswan} R package as an
example dataset for reproducing the sliding-window methodology
of the source paper.

We treat the analysis as an illustrative proof-of-concept
application of the proposed test, not as a definitive biomarker
discovery analysis. The reported counts are exploratory,
computed at the nominal $5\%$ level without adjustment for
multiplicity; we discuss multiplicity in
Section~\ref{sec:aging-results}. We use this publicly available
subset for our application to ensure that every number in this
section can be reproduced from a single R script
(Web Appendix~G) that downloads the data from the package
repository and runs end-to-end with no data-access barriers.

While the DEswan subset is much smaller than the source study,
it is a well-characterized publicly available dataset drawn from
the study of \citet{lehallier2019undulating}. Its four
constituent cohorts differ substantially in age range and
provide, in the Seattle cohort ($n = 70$, ages $21$--$88$), a
single-site sample with the widest single-cohort age span in
the data. The sample size ($n = 70$) is slightly below the
lower end of the paper's target operating range of
$n \approx 100$--$500$ used in the simulations
(Section~\ref{sec:simulations}), which contributes to the
descriptive character of the analysis reported below.

\subsection{Choice of study population}\label{sec:aging-cohort}

The four cohorts in the DEswan subset span very different age
ranges (Seattle: $21$--$88$, PRIN06: $50$--$102$, PRIN09:
$56$--$107$, GEHA: $89$--$104$), and GEHA is additionally a
selected sample of long-lived individuals. Because cohorts
differ in collection site, processing pipeline and subject
selection, protein levels differ systematically between
cohorts at matched ages, and a marginal test of independence
between age and protein abundance across the pooled cohorts
would confound age dependence with between-cohort batch
differences.

We therefore restrict the analysis to the Seattle cohort
($n = 70$): a single collection site with a single processing
pipeline, spanning the widest single-cohort age range
($21$--$88$ years). This choice removes the between-cohort
confounding described above at the cost of reduced sample size.
We do not attempt a pooled analysis of the full $n = 171$
sample, because any between-cohort adjustment would test a
different question and its interpretation would depend on the
model for the cohort effect.

\subsection{Testing setup}\label{sec:aging-setup}

For each of the $1{,}305$ plasma proteins we test independence
between chronological age and log protein abundance in the
Seattle cohort. Four methods are compared: Pearson $r$,
Spearman $\rho$, Chatterjee's $\xi_{\text{Ch}}$, and the
proposed $T_{\mathrm{cut}}$. All four use exact permutation
calibration with $B = 999$ replicates.

Age is recorded in integer years and takes only $47$ distinct
values among the $70$ Seattle subjects. Chatterjee's
$\xi_{\text{Ch}}$ requires that ties in $X$ be broken at
random; we implement this by ordering the sample on $X$ with
ties resolved uniformly by a random permutation of tied
positions, redrawing the tie-breaking permutation on each
call to the statistic (both the observed value and the
permutation replicates). This matches the standard
implementation in the \texttt{XICOR} R package.

\subsection{Results: a complementary picture}\label{sec:aging-results}

Classical rank-based coefficients and the proposed
$T_{\mathrm{cut}}$ target different aspects of the
age--protein dependence structure. Pearson's $r$ detects linear
trends in mean protein abundance; Spearman's $\rho$ detects
monotone trends in the ranks; Chatterjee's $\xi_{\text{Ch}}$
detects functional dependence with uniform threshold weighting;
$T_{\mathrm{cut}}$ detects any change in the full
conditional distribution of protein abundance with age, with
Fisher-type emphasis across thresholds. In a biomarker-screening
application, the four coefficients therefore answer four
distinct scientific questions, and are best used in combination.

Table~\ref{tab:aging} reports rejection counts under both
unadjusted $p < 0.05$ and Benjamini--Hochberg (BH) FDR control
at $q < 0.05$. Under BH control $T_{\mathrm{cut}}$ rejects on
$70$ proteins ($5.4\%$), Pearson on $112$ ($8.6\%$), Spearman
on $142$ ($10.9\%$), and Chatterjee $\xi_{\text{Ch}}$ on none
of the panel. The BH-controlled sets overlap substantially but
are not nested: six proteins are rejected by $T_{\mathrm{cut}}$
at $q < 0.05$ but by none of the three classical methods at the
same FDR level. Under the more liberal unadjusted $5\%$ level,
$T_{\mathrm{cut}}$ rejects on $207$ proteins, of which $24$ are
missed by all three classical methods.

That Chatterjee's $\xi_{\text{Ch}}$ reports no BH-controlled
rejections is a consequence of its permutation calibration
across $m = 1{,}305$ tests: at $B = 999$ the minimum permutation
$p$-value is $0.001$, and at the smallest BH threshold
$0.05 / m \approx 3.8 \times 10^{-5}$ no single rejection can
survive. $T_{\mathrm{cut}}$ recovers BH-controlled findings on
the same permutation grid by concentrating enough of its
$p$-values close to the minimum: it has $70$ proteins with
$p$-values below the BH threshold at $q < 0.05$.

The pattern of non-rejection by the classical methods does not
by itself identify the shape of the alternative---the classical
procedures may lack power against weak monotone dependence as
well as against non-monotone alternatives---but the overlap
structure demonstrates that no single coefficient dominates on
this dataset.

Figure~\ref{fig:aging} shows all six proteins on which
$T_{\mathrm{cut}}$ rejects under BH control at $q < 0.05$ while
none of the three classical methods reject at the same FDR
level. Panels (a)--(f) show Feature\_167, Feature\_376,
Feature\_516, Feature\_523, Feature\_533, and Feature\_1260. In
every panel the ordinary least-squares fit is essentially flat
(the fitted slope explains negligible variance), while the LOESS
smoother traces distributional structure across age that is not
captured by the location trend. For $T_{\mathrm{cut}}$,
per-protein permutation $p$-values are $\le 0.002$; the
BH-adjusted $q$-values are in the range $0.019$--$0.037$. Some
of the six proteins have a single classical unadjusted $p$-value
below $0.05$ (e.g., Spearman for Feature\_167, Feature\_376 and
Feature\_1260), but none survives BH control across the
$m = 1{,}305$-test panel. We caution that at $n = 70$ LOESS
wiggles are noisy, and the visual patterns should be read as
illustrative of the structure $T_{\mathrm{cut}}$ detects rather
than as definitive shape claims for individual proteins.

\subsection{Interpretation}\label{sec:aging-interp}

The $6$ proteins where $T_{\mathrm{cut}}$ rejects under
BH control at $q < 0.05$ but the three classical methods do not
(and the further $24$ proteins where $T_{\mathrm{cut}}$
nominally rejects at $p < 0.05$ while the classical methods do
not) exhibit dependence patterns that classical rank- and
correlation-based coefficients did not surface in this scan.
Several show visibly non-monotone age patterns; whether any of
these represent \emph{life-stage-specific} biomarkers---proteins
whose plasma abundance changes shape at particular life stages
such as reproductive maturity, mid-life metabolic transitions,
or the onset of senescence---would require independent
replication in a larger cohort with, ultimately, biological
follow-up. The larger source study
\citep{lehallier2019undulating} identified non-linear waves of
proteomic change at approximately ages $34$, $60$ and $78$; the
Seattle subset used here contains too few subjects under age
$50$ to detect those specific waves at the individual-protein
level with adequate power, and we do not attempt to do so.

More broadly, the analysis illustrates that a complete
biomarker screen benefits from applying multiple dependence
coefficients targeting different aspects of the joint
distribution: linear methods for aging clocks, monotone rank
tests for progression markers, and distributional tests such as
$T_{\mathrm{cut}}$ for transition markers. Each
method contributes a distinct set of hits, and the union is
strictly larger than any single method's output.

\section{Discussion}\label{sec:discussion}

The composite Bernoulli likelihood-ratio formulation developed
here provides a route to nonparametric independence testing that
is anchored in likelihood principles while avoiding the
instability of conditional-density estimation. Classical
dependence measures---Pearson's $r$, Spearman's $\rho$, and
Kendall's $\tau$---target linear, monotone, or concordant
association and can vanish under pronounced non-monotone
dependence. Generalizing the parametric likelihood ratio to a
binomial-cut composite likelihood ratio yields a population
functional that characterizes independence, with the Fisher
weighting of Anderson--Darling-type statistics emerging at the
second order as a natural consequence of the Bernoulli
likelihood rather than an external design choice. On biomedical data with variance and tail dependence, the
resulting permutation test detects associations that rank-based
coefficients miss.

\section*{Data Availability Statement}

The aging plasma proteome dataset used in Section~\ref{sec:aging}
is publicly available in the DEswan R package repository at
\url{https://github.com/lehallib/DEswan} \citep{shen2020deswan}.
A single self-contained R script reproducing all numerical
results in Section~\ref{sec:aging} and Table~\ref{tab:aging},
including automated data download, is provided as part of the
online Supplementary Material. A companion reproducibility
bundle including this script, the simulation verification
script, the figure source, and all manuscript files is
available upon request from the corresponding author.

\section*{AI Use Statement}

Anthropic's Claude was used as an editorial and drafting
assistant during manuscript preparation, including for
copy-editing, wording suggestions, LaTeX debugging, response
letter drafting, and cross-reference checking. All statistical
methodology, mathematical derivations, simulation code,
biomedical data analysis, and scientific claims are those of
the author, who verified every generated numerical result and
takes full responsibility for the content of the manuscript.

\section*{Supplementary Materials}

Web Appendices A--G contain proofs of the population properties
(A), the second-order equivalence with $\xi_{AD}$ (B), the
permutation-exactness argument (C), an information-theoretic
interpretation and lower bound on $I(X; Y)$ (D), interpretations
under location-shift, location-scale, and Cox proportional
hazards models (E), additional simulation results and
sensitivity analyses (F), and computational details of the
aging-proteome application together with the reproduction
script (G).

\section*{Funding Statement}

This research was supported in part by the Intramural Research
Program of the National Institutes of Health (NIH). The
contributions of the NIH author(s) are considered Works of the
United States Government. The findings and conclusions
presented in this paper are those of the author(s) and do not
necessarily reflect the views of the NIH or the U.S.\
Department of Health and Human Services.

\begingroup
\small
\setlength{\bibsep}{4pt}

\endgroup


\begin{table}[!ht]
\centering
\renewcommand{\arraystretch}{1.1}
\caption{Empirical power at nominal level $\alpha = 0.05$ for
six scenarios at $n = 200$ and $\lambda = 0.5$ (moderate
signal-to-noise), based on $1000$ Monte Carlo replications with
$B = 999$ permutations. Under (S1) the entry is empirical size.
Bold indicates the two-largest values in each row (excluding
size). Analogous results at $n \in \{100, 500\}$ and
$\lambda \in \{0.25, 0.75\}$ are in Web Appendix~F.}
\label{tab:power}
\begin{tabular}{lccccc}
\hline
Scenario
  & $\xi_{\text{Ch}}$
  & $\nu_{\text{AR}}$
  & $\xi_{\text{Ch,boost}}^{M=50}$
  & $\widehat\xi_{AD}^*$
  & $T_{\mathrm{cut}}$ \\
\hline
S1 Null                    & 0.041 & 0.051 & 0.043 & 0.049 & 0.051 \\
S2 Linear                  & 0.112 & 0.121 & \textbf{0.754} & \textbf{0.396} & 0.386 \\
S3 Threshold               & 0.989 & 0.984 & \textbf{1.000} & \textbf{1.000} & \textbf{1.000} \\
S4 Symmetric non-monotone  & 0.100 & 0.118 & 0.048 & \textbf{0.327} & \textbf{0.321} \\
S5 Oscillating             & \textbf{1.000} & \textbf{1.000} & 0.000 & \textbf{1.000} & \textbf{1.000} \\
S6 Heteroscedastic         & 0.090 & 0.137 & 0.258 & \textbf{0.446} & \textbf{0.420} \\
\hline
\end{tabular}
\end{table}

\begin{table}[!ht]
\centering
\renewcommand{\arraystretch}{1.1}
\caption{Scan of $1{,}305$ plasma proteins for age
dependence on the Seattle cohort of the aging plasma proteome
dataset \citep{shen2020deswan, lehallier2019undulating}
($n = 70$ subjects, ages $21$--$88$; single collection site).
All four methods use exact permutation calibration with
$B = 999$ replicates. The four methods target
different aspects of the joint distribution and reject on
overlapping but non-nested sets of proteins. Under
Benjamini--Hochberg control at $q < 0.05$,
$T_{\mathrm{cut}}$ rejects on $70$ proteins ($5.4\%$),
of which $6$ are missed by all three classical methods at the
same FDR level. Figure~\ref{fig:aging} illustrates the
$T_{\mathrm{cut}}$-versus-classical complementarity on all six
of these $T_{\mathrm{cut}}$-only rejections.}
\label{tab:aging}
\begin{tabular}{lcccc}
\hline
Method & \multicolumn{2}{c}{Nominal $p<0.05$}
       & \multicolumn{2}{c}{BH $q<0.05$} \\
       & Count & \% & Count & \% \\
\hline
Pearson correlation ($r$)     & 249 & 19.1\% & 112 & \phantom{0}8.6\% \\
Spearman correlation ($\rho$) & 289 & 22.1\% & 142 & 10.9\% \\
Chatterjee $\xi_{\text{Ch}}$  & 130 & 10.0\% & \phantom{00}0 & \phantom{0}0.0\% \\
$T_{\mathrm{cut}}$ (proposed) & 207 & 15.9\% & \phantom{0}70 & \phantom{0}5.4\% \\
\hline
\multicolumn{5}{l}{\emph{$T_{\mathrm{cut}}$ rejects but the alternative does not:}} \\
                              & \multicolumn{2}{c}{Nominal $p<0.05$} & \multicolumn{2}{c}{BH $q<0.05$} \\
vs.\ Pearson                  & \multicolumn{2}{c}{63} & \multicolumn{2}{c}{13} \\
vs.\ Spearman                 & \multicolumn{2}{c}{43} & \multicolumn{2}{c}{\phantom{0}6} \\
vs.\ Chatterjee $\xi_{\text{Ch}}$ & \multicolumn{2}{c}{106} & \multicolumn{2}{c}{70} \\
vs.\ all three classical methods  & \multicolumn{2}{c}{\textbf{24}} & \multicolumn{2}{c}{\textbf{6}} \\
\hline
\end{tabular}
\end{table}

\begin{figure}[!ht]
\centering
\includegraphics[width=0.98\textwidth]{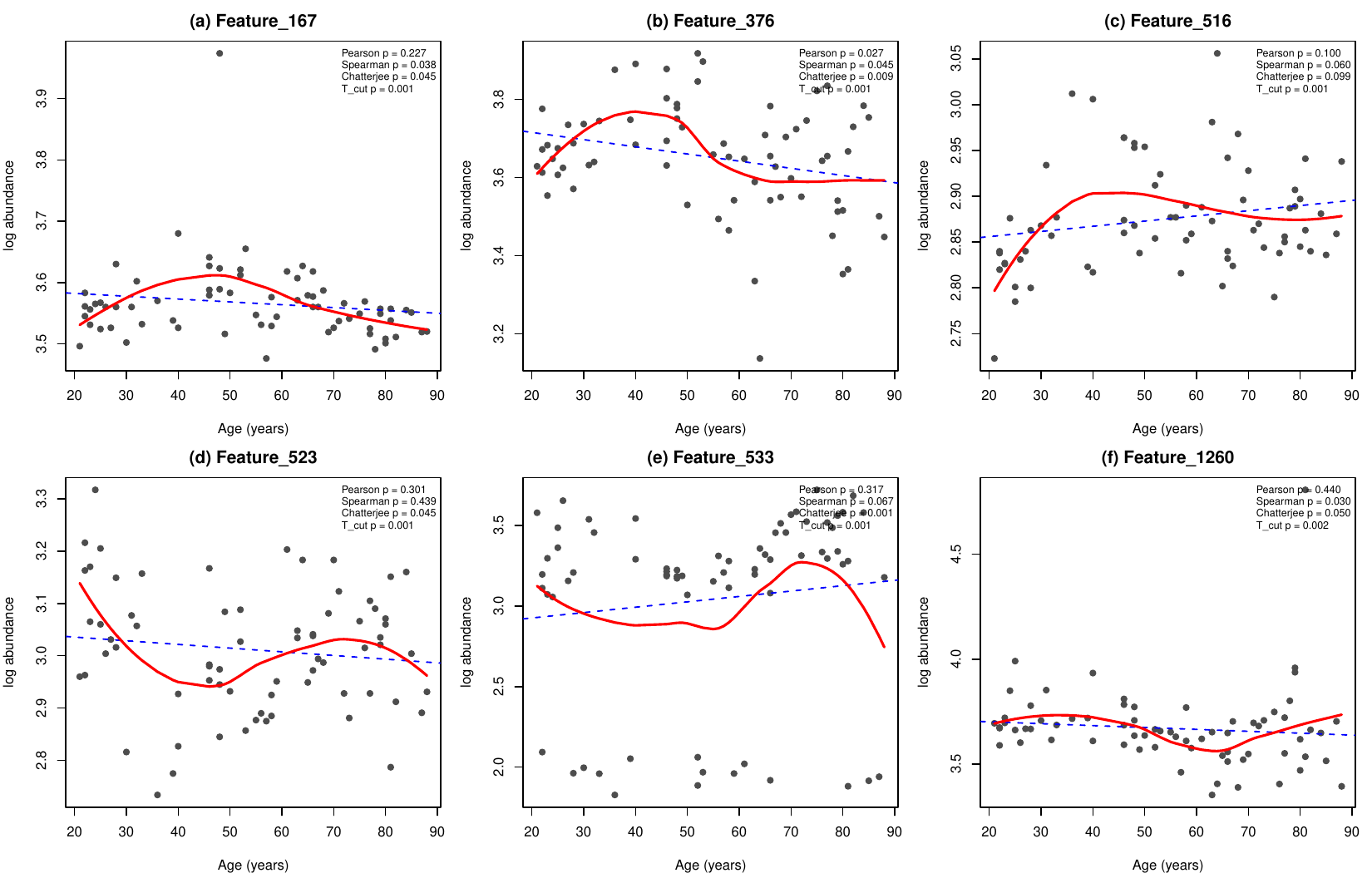}
\caption{The six plasma proteins from the Seattle cohort of the
aging proteome dataset \citep{shen2020deswan,
lehallier2019undulating} on which $T_{\mathrm{cut}}$ rejects
under Benjamini--Hochberg control at $q < 0.05$ while none of
the three classical methods (Pearson, Spearman, Chatterjee
$\xi_{\text{Ch}}$) reject at the same FDR level. Grey points
are the $n = 70$ Seattle observations; blue dashed lines are
ordinary least-squares fits (essentially flat in all six
panels); red solid curves are LOESS smoothers, which trace the
distributional structure that $T_{\mathrm{cut}}$ detects.
Per-panel unadjusted $p$-values from each of the four methods
(all four from $B = 999$ permutation calibration) are shown in
the top right of each panel. Some proteins have a single
classical $p$-value below $0.05$ (e.g., Spearman for
Feature\_167, Feature\_376 and Feature\_1260), but none of
these survives Benjamini--Hochberg control across
$m = 1{,}305$ tests whereas $T_{\mathrm{cut}}$ does. The
BH-adjusted $q$-values for $T_{\mathrm{cut}}$ on these six
proteins range from $0.019$ to $0.037$.}
\label{fig:aging}
\end{figure}

\clearpage
\appendix

\setcounter{section}{0}
\renewcommand{\thesection}{\Alph{section}}
\renewcommand{\thesubsection}{\Alph{section}.\arabic{subsection}}
\renewcommand{\theequation}{\Alph{section}.\arabic{equation}}
\renewcommand{\thefigure}{S\arabic{figure}}
\renewcommand{\thetable}{S\arabic{table}}
\setcounter{equation}{0}
\setcounter{figure}{0}
\setcounter{table}{0}

\begin{center}
{\Large\bfseries Supplementary Material}\\[0.5em]
Web Appendices A--G
\end{center}
\bigskip

\section*{Notation}
\addcontentsline{toc}{section}{Notation}

This supplement uses the following notation, consistent with the
main paper:

\begin{itemize}\setlength\itemsep{2pt}
  \item $F(t) = \Pr(Y \le t)$: marginal CDF of $Y$;
  \item $F(t \mid x) = \Pr(Y \le t \mid X = x)$: conditional CDF;
  \item $\phi(u, v) = u \log(u/v) + (1-u)\log((1-u)/(1-v))$:
        Bernoulli KL divergence;
  \item $\phi_\varepsilon(u, v) = (u+\varepsilon) \log\!\bigl((u+\varepsilon)/(v+\varepsilon)\bigr)
        + (1-u+\varepsilon) \log\!\bigl((1-u+\varepsilon)/(1-v+\varepsilon)\bigr)$:
        $\varepsilon$-regularized version;
  \item $M(v) = v \log(1/v) + (1-v)\log(1/(1-v))$: upper bound on $E[\phi(U, v)]$ under $E[U] = v$, $U \in [0,1]$,
        with $\int_0^1 M(v)\, dv = 1/2$;
  \item $M_\varepsilon(v) = v\, \phi_\varepsilon(1, v) + (1-v)\, \phi_\varepsilon(0, v)$:
        analogous averaged upper bound for $\phi_\varepsilon$;
  \item $\xi_{\mathrm{cut}}(X,Y) = 2 \int \E_X[\phi(F(t\mid X), F(t))]\, dF(t)$:
        the population coefficient (main paper equation (4));
  \item $\xi_{AD}(X,Y) = \int \E_X[(F(t\mid X) - F(t))^2] / (F(t)(1-F(t)))\, dF(t)$:
        the Fisher-weighted $L^2$ counterpart (main paper equation (5)).
\end{itemize}

Results in this supplement are labelled with an S prefix (Theorem S1,
Lemma S1.2, etc.) to distinguish them from results in the main paper.

Throughout, we assume $Y$ has a continuous marginal distribution
function $F$ with interval support, and assume $F$ is strictly
increasing on that support. This holds whenever $Y$ has a density
that is positive on its support, which is the standard regularity
in biomarker practice. Under this condition, $F$ is atomless and
the probability integral transform $F(Y) \sim \mathrm{Uniform}[0, 1]$
is exact, so the change of variable $\int g(F(t))\, dF(t) =
\int_0^1 g(v)\, dv$ holds for any measurable $g$. This is the
same assumption as in the main paper's Section~2.1, made explicit
here for the proofs.

\section{Proofs of the population properties (Theorem 1)}

This section proves the four parts of Theorem~2.1 of the main paper.
We begin with a technical lemma collecting properties of the
Bernoulli KL divergence.

\subsection{Properties of \texorpdfstring{$\phi$}{phi}}

\begin{slemma}[Properties of $\phi$]
\label{lem:phi-props}
For every $v \in (0, 1)$:
\begin{enumerate}[label=\textup{(\roman*)},leftmargin=2em]
  \item (Convexity and nonnegativity.) $u \mapsto \phi(u, v)$
        is strictly convex on $(0, 1)$ with $\phi(u, v) \ge 0$,
        equality iff $u = v$.
  \item (Symmetry.) $\phi(u, v) = \phi(1-u, 1-v)$
        and $M(v) = M(1-v)$.
  \item (Upper envelope.) If $\E(U) = v$ and $U \in [0, 1]$
        almost surely, then
        $\E[\phi(U, v)] \le M(v)$, with equality iff
        $U \in \{0, 1\}$ a.s.\ with $\Pr(U = 1) = v$.
  \item (Entropy identity.)
        $\int_0^1 M(v)\, dv = 1/2$.
\end{enumerate}
The regularized statements for $\phi_\varepsilon$: parts (i)--(iii)
and the boundedness envelope
$\sup_{u, v \in [0, 1]} |\phi_\varepsilon(u, v)|
\le 2(1 + \varepsilon)\log((1 + \varepsilon)/\varepsilon) < \infty$
hold analogously. The normalization identity in part (iv) is
specific to $\varepsilon = 0$: for $\varepsilon > 0$,
$\int_0^1 M_\varepsilon(v)\, dv < 1/2$, with equality in the limit
$\varepsilon \downarrow 0$. For example,
$\int_0^1 M_{0.01}(v)\, dv \approx 0.4634$.
\end{slemma}

\begin{proof}
\emph{Part (i).} Differentiating $\phi(u, v)$ twice with respect
to $u$ gives
\[
\frac{\partial^2}{\partial u^2} \phi(u, v)
= \frac{1}{u} + \frac{1}{1 - u} > 0,
\qquad u \in (0, 1),
\]
so $\phi(\cdot, v)$ is strictly convex on $(0, 1)$. Since
$\phi(v, v) = 0$ and $\partial_u \phi(u, v)|_{u = v} = 0$, the point
$u = v$ is the unique minimizer and $\phi(u, v) \ge 0$ with equality
iff $u = v$.

\emph{Part (ii).} From the definition,
\[
\phi(1 - u, 1 - v)
= (1 - u) \log\!\frac{1 - u}{1 - v}
+ u \log\!\frac{u}{v}
= \phi(u, v).
\]
Then $M(1 - v) = (1-v)\phi(1, 1-v) + v \phi(0, 1-v)
= (1-v)\phi(0, v) + v \phi(1, v) = M(v)$.

\emph{Part (iii).} By convexity of $\phi(\cdot, v)$ on $[0, 1]$,
the chord inequality gives
\[
\phi(u, v) \le (1 - u)\, \phi(0, v) + u\, \phi(1, v),
\qquad 0 \le u \le 1.
\]
Taking expectations of both sides with $\E U = v$,
\[
\E[\phi(U, v)]
\le (1 - v)\, \phi(0, v) + v\, \phi(1, v)
= M(v).
\]
By strict convexity of $\phi(\cdot, v)$, the chord inequality is
strict whenever $u \in (0, 1)$. Therefore equality in the
integrated form requires $U \in \{0, 1\}$ almost surely, in which
case $\E U = v$ forces $\Pr(U = 1) = v$ and $\Pr(U = 0) = 1 - v$.

\emph{Part (iv).}
$\int_0^1 M(v)\, dv = -2 \int_0^1 v \log v\, dv$
$= 2 \cdot [-v^2/2 \cdot \log v + v^2/4]_0^1$
$= 1/2$.

Parts~(i)--(iii) for $\phi_\varepsilon$ follow by the same
arguments; part~(iv) does not, since
$\int_0^1 M_\varepsilon(v)\, dv$ depends on $\varepsilon$ (see the
statement of the lemma).
\end{proof}

\subsection{Range}

\begin{proof}[Proof of Theorem 1(i)]
Nonnegativity is immediate from
Lemma~\ref{lem:phi-props}(i): $\phi(F(t\mid X), F(t)) \ge 0$ pointwise,
so the integral is nonnegative.

For the upper bound, apply Lemma~\ref{lem:phi-props}(iii) with
$U_t := F(t \mid X)$: since $\E[U_t] = F(t)$ and $U_t \in [0, 1]$,
$\E[\phi(U_t, F(t))] \le M(F(t))$. Integrating and using
Lemma~\ref{lem:phi-props}(iv),
\[
\xi_{\mathrm{cut}}
= 2 \int \E[\phi(F(t\mid X), F(t))]\, dF(t)
\le 2 \int M(F(t))\, dF(t)
= 2 \int_0^1 M(v)\, dv
= 1,
\]
where the change of variable $v = F(t)$ used continuity of $F$
(so the pushforward of $dF$ is $\text{Uniform}[0,1]$).
\end{proof}

\subsection{Preliminary: support containment}

Both the independence characterization (Theorem 1(ii)) and the
functional-dependence characterization (Theorem 1(iii)) use the
following basic fact: the conditional distribution of $Y$ given
$X$ is almost surely supported inside the support of the marginal
distribution of $Y$. We state this as a lemma for later
reference.

\begin{slemma}[Support containment]
\label{lem:support-containment}
Let $Y$ be a real-valued random variable with distribution $\nu$
(so $\nu((-\infty, t]) = F(t)$), and let $\nu_x$ denote the
regular conditional distribution of $Y$ given $X = x$. Then
$\mathrm{supp}(\nu_x) \subseteq \mathrm{supp}(\nu)$ for
$\Pr_X$-almost every $x$.
\end{slemma}

\begin{proof}
Let $\mathcal{I}$ denote the countable family of open intervals
in $\R$ with rational endpoints. For any $I \in \mathcal{I}$ with
$\nu(I) = \Pr(Y \in I) = 0$, the tower property gives
\[
0 \;=\; \Pr(Y \in I) \;=\; \E[\nu_X(I)],
\]
so $\nu_x(I) = 0$ for $\Pr_X$-almost every $x$. Taking the
intersection over the countable family
$\{I \in \mathcal{I} : \nu(I) = 0\}$, there is a set $C$ with
$\Pr_X(C) = 1$ such that $\nu_x(I) = 0$ for every such $I$ and
every $x \in C$. Since every open subset of
$\mathrm{supp}(\nu)^c$ is a countable union of rational open
intervals of $\nu$-measure zero, this gives
$\nu_x(\mathrm{supp}(\nu)^c) = 0$ for every $x \in C$, i.e.,
$\mathrm{supp}(\nu_x) \subseteq \mathrm{supp}(\nu)$.
\end{proof}

\subsection{Characterization of independence}

\begin{proof}[Proof of Theorem 1(ii):
$\xi_{\mathrm{cut}} = 0 \iff X \indep Y$]
The reverse direction is immediate: if $X \indep Y$, then
$F(t \mid X) = F(t)$ almost surely for every $t$, and
$\phi(F(t), F(t)) = 0$.

For the forward direction, suppose $\xi_{\mathrm{cut}} = 0$. Since
$\phi \ge 0$ (Lemma~\ref{lem:phi-props}(i)), the integrand vanishes
almost surely under the product measure $F \otimes \Pr_X$: there
exists a set $N \subset \R \times \Omega$ with $(F \otimes \Pr)(N) = 0$
such that $F(t \mid X(\omega)) = F(t)$ for all $(t, \omega) \notin N$.
By Fubini, there is a set $\Omega_0$ of full $\Pr_X$-measure such that,
for every $\omega \in \Omega_0$,
\[
F(t \mid X(\omega)) = F(t)
\qquad \text{for $F$-almost every $t$.}
\]

We now extend to \emph{all} $t$. Under the interval-support
assumption, every nonempty relatively open interval in
$\mathrm{supp}(F)$ has positive $F$-measure, so the complement
of an $F$-null set is dense in $\mathrm{supp}(F)$. Both $t \mapsto F(t \mid X)$ and
$t \mapsto F(t)$ are right-continuous. For any interior point
$t_0$ of $\mathrm{supp}(F)$, take a sequence $t_n \downarrow t_0$
inside $\mathrm{supp}(F)$ on which
$F(t_n \mid X) = F(t_n)$, then right-continuity gives
$F(t_0 \mid X) = F(t_0)$. At the right endpoint $b$ of
$\mathrm{supp}(F)$ (if finite), $F(b) = 1$ and by
Lemma~\ref{lem:support-containment} $F(b \mid X) = 1$ almost
surely, so equality holds at $b$ as well. For $t$ outside
$\mathrm{supp}(F)$ the values of $F$ and $F(\cdot \mid X)$ are
constant $0$ or $1$ on the components of the complement, and
equality follows by monotonicity from the endpoints. Therefore
$F(t \mid X) = F(t)$ almost surely for every $t \in \R$, which is
equivalent to $X \indep Y$.
\end{proof}

\subsection{Characterization of functional dependence}

\begin{proof}[Proof of Theorem 1(iii):
$\xi_{\mathrm{cut}} = 1 \iff Y = g(X)$ a.s.\ for some measurable $g$]
\emph{Reverse direction.} If $Y = g(X)$ almost surely,
$F(t \mid X) = \mathbf{1}\{g(X) \le t\} \in \{0, 1\}$ for every $t$,
so the equality condition in Lemma~\ref{lem:phi-props}(iii) holds
for every $t$: $\E[\phi(U_t, F(t))] = M(F(t))$. Integrating and
using Lemma~\ref{lem:phi-props}(iv) gives $\xi_{\mathrm{cut}} = 1$.

\emph{Forward direction.} Suppose $\xi_{\mathrm{cut}} = 1$.
By Lemma~\ref{lem:phi-props}(iii), equality
$\E[\phi(U_t, F(t))] = M(F(t))$ holds for $F$-almost every $t$,
which by the equality condition gives
\begin{equation}\label{eq:UtBernoulli}
F(t \mid X) \in \{0, 1\}\quad\text{a.s.\ for $F$-a.e.\ }t.
\end{equation}

By Fubini and the interval-support assumption, there exists a
measurable set $B$ with $\Pr_X(B) = 1$ such that, for every
$x \in B$, \eqref{eq:UtBernoulli} holds. Fix such $x$ and let
$\nu_x$ be the conditional distribution of $Y$ given $X = x$. By
Lemma~\ref{lem:support-containment}, there exists $C$ with
$\Pr_X(C) = 1$ such that
$\mathrm{supp}(\nu_x) \subseteq \mathrm{supp}(F)$ for every
$x \in C$.

\emph{Point-mass conclusion.} Fix $x \in B \cap C$. By
\eqref{eq:UtBernoulli}, $F(t \mid x) \in \{0, 1\}$ for $F$-a.e.\
$t$. Under the interval-support assumption, every nonempty
relatively open interval in $\mathrm{supp}(F)$ has positive
$F$-measure (note: this does \emph{not} require $F$ to be
absolutely continuous with respect to Lebesgue measure). Hence
the complement of the $F$-null set on which
$F(\cdot \mid x) \in \{0, 1\}$ fails is dense in
$\mathrm{supp}(F)$. Right-continuity of $F(\cdot \mid x)$ extends
this to $F(t \mid x) \in \{0, 1\}$ at every interior point
$t_0 \in \mathrm{supp}(F)$: take $t_n \downarrow t_0$ inside
$\mathrm{supp}(F)$ with $F(t_n \mid x) \in \{0, 1\}$; then
$F(t_0 \mid x) \in \{0, 1\}$ since $\{0, 1\}$ is closed under the
right limit of a $\{0, 1\}$-valued sequence. At the right endpoint
$b$ of $\mathrm{supp}(F)$ (if finite), $F(b \mid x) = 1$ by
Lemma~\ref{lem:support-containment} applied at $x \in C$. Thus
$F(\cdot \mid x)$ is $\{0, 1\}$-valued on $\mathrm{supp}(F)$.

A non-decreasing right-continuous function taking only the values
$0$ and $1$ on a connected support has exactly one jump from $0$
to $1$, at a point $g(x) \in \mathrm{supp}(F)$. Measurability of
$x \mapsto g(x)$ follows from measurability of
$x \mapsto F(t \mid x)$ for each $t$, via $g(x) = \inf\{t \in \Q :
F(t \mid x) = 1\}$ (using a countable rational grid). Then
$\nu_x = \delta_{g(x)}$, so $F(t \mid X) = \mathbf{1}\{g(X) \le t\}$
almost surely, i.e., $Y = g(X)$ almost surely.
\end{proof}

\subsection{Invariance under monotone transformations}

\begin{proof}[Proof of Theorem 1(iv):
$\xi_{\mathrm{cut}}$ is invariant under strictly monotone
transformations of $Y$]
Let $Y^* = h(Y)$ with $h$ strictly monotone.

\emph{Case 1: $h$ strictly increasing.} Then
$\Pr(Y^* \le s \mid X) = F(h^{-1}(s) \mid X)$ and
$\Pr(Y^* \le s) = F(h^{-1}(s))$. Under the change of variable
$t = h^{-1}(s)$, we have $dF_{Y^*}(s) = dF(t)$, so
\[
2 \int \E\bigl[\phi(F(h^{-1}(s) \mid X), F(h^{-1}(s)))\bigr]\, dF_{Y^*}(s)
= 2 \int \E\bigl[\phi(F(t \mid X), F(t))\bigr]\, dF(t)
= \xi_{\mathrm{cut}}(X, Y).
\]

\emph{Case 2: $h$ strictly decreasing.} Then
$\{h(Y) \le s\} = \{Y \ge h^{-1}(s)\}$, and
$\Pr(Y^* \le s \mid X) = 1 - F(h^{-1}(s)^- \mid X)$, where
$F(t^-)$ denotes the left-limit.

Continuity of the marginal $F$ does not by itself imply
continuity of $F(\cdot \mid X)$ pointwise: conditional
distributions can have atoms even when the marginal is
continuous (for example, if $Y = X$, then $Y \mid X = x$ is a
point mass despite $Y$ having a continuous marginal). What holds
under continuity of $F$ is the following: for each fixed $t$,
\[
\E[\Pr(Y = t \mid X)] = \Pr(Y = t) = 0,
\]
so $\Pr(Y = t \mid X) = 0$ almost surely, hence
\begin{equation}
\label{eq:leftlim-fixed-t}
F(t^- \mid X) = F(t \mid X)\quad\text{almost surely, for each fixed } t.
\end{equation}
The exceptional null set depends on $t$, but this is enough:
substituting \eqref{eq:leftlim-fixed-t} inside the integral over
$t \sim F$ and applying Fubini yields
\[
\int \E[\phi(\Pr(Y^* \le s \mid X), \Pr(Y^* \le s))]\, dF_{Y^*}(s)
= \int \E[\phi(1 - F(t \mid X), 1 - F(t))]\, dF(t).
\]
By Lemma~\ref{lem:phi-props}(ii),
$\phi(1 - u, 1 - v) = \phi(u, v)$, so the integrand is unchanged
under the substitution and $\xi_{\mathrm{cut}}(X, Y^*) =
\xi_{\mathrm{cut}}(X, Y)$.
\end{proof}

\section{Second-order equivalence with \texorpdfstring{$\xi_{AD}$}{xi\_AD}}

The following result formalizes the Taylor-expansion argument used
in Section~3.1 of the main paper.

\begin{slemma}[Taylor expansion of $\phi$]
\label{lem:taylor-phi}
For every $u, v \in (0, 1)$,
\[
\phi(u, v)
= \frac{(u - v)^2}{2\, v(1 - v)}
+ R(u, v),
\]
where the remainder satisfies
\[
|R(u, v)| \le \frac{1}{6\,(\min\{u, v\})^2 \, (\min\{1 - u, 1 - v\})^2}\,
|u - v|^3
\qquad\text{for } u, v \in (0, 1).
\]
For the regularized $\phi_\varepsilon$, the leading quadratic
term is not $(u - v)^2 / (2[v(1-v) + \varepsilon])$ as one might
guess from a naive substitution. Direct computation gives
\[
\frac{\partial^2 \phi_\varepsilon}{\partial u^2}\bigg|_{u = v}
= \frac{1}{v + \varepsilon} + \frac{1}{1 - v + \varepsilon}
= \frac{1 + 2\varepsilon}{(v + \varepsilon)(1 - v + \varepsilon)}
= \frac{1 + 2\varepsilon}{v(1 - v) + \varepsilon + \varepsilon^2},
\]
so the actual quadratic term is
\[
\phi_\varepsilon(u, v)
= \frac{1 + 2\varepsilon}{2 \bigl[v(1-v) + \varepsilon + \varepsilon^2\bigr]}
  (u - v)^2 + O((u - v)^3).
\]
At $\varepsilon = 0$ this recovers $(u-v)^2/(2 v(1-v))$; for
$\varepsilon > 0$ the correct denominator is
$v(1-v) + \varepsilon + \varepsilon^2$ and the numerator gains a
factor of $1 + 2\varepsilon$.
\end{slemma}

\begin{proof}
The Taylor expansion is standard. Direct computation gives
$\partial_u \phi(u, v)|_{u = v} = 0$ and
$\partial_u^2 \phi(u, v)|_{u = v} = 1/(v(1-v))$, so
$\phi(u, v) = (u - v)^2/(2 v(1-v)) + R$ with $R$ given by Lagrange
remainder. The third-derivative bound gives the stated bound on $R$.
The regularized case is a direct calculation.
\end{proof}

Substituting the pointwise expansion of Lemma~\ref{lem:taylor-phi}
into the definition of $\xi_{\mathrm{cut}}$ and dividing by 2 gives
the exact integrated identity
\begin{equation}
\label{eq:xicut-xiAD-exact}
\xi_{\mathrm{cut}}(X, Y) - \xi_{AD}(X, Y)
= 2 \int \E_X\!\bigl[R(F(t \mid X), F(t))\bigr]\, dF(t),
\end{equation}
where $R(u, v) = \phi(u, v) - (u - v)^2 / (2 v(1 - v))$ is the
third-order remainder in Lemma~\ref{lem:taylor-phi}. Using the
cubic bound in Lemma~\ref{lem:taylor-phi} on $|R|$, one has the
non-uniform bound
\begin{equation}
\label{eq:xicut-xiAD-bound}
\bigl|\xi_{\mathrm{cut}}(X, Y) - \xi_{AD}(X, Y)\bigr|
\le \frac{1}{3} \int \E_X\!\left[
\frac{|F(t \mid X) - F(t)|^3}{a_t^2\, b_t^2}
\right] dF(t),
\end{equation}
where
$a_t = \min\{F(t \mid X), F(t)\}$ and
$b_t = \min\{1 - F(t \mid X), 1 - F(t)\}$
(with the convention $|R| = 0$ when the denominator vanishes).

The bound is non-uniform because $a_t$ and $b_t$ may approach
zero in the tails of $F$. Consequently, local second-order
equivalence between $\xi_{\mathrm{cut}}$ and $\xi_{AD}$ holds
along sequences of alternatives for which the weighted cubic
remainder in \eqref{eq:xicut-xiAD-bound} is $o$ of the
corresponding quadratic term $\xi_{AD}$. This requires suitable
tail or uniform-integrability conditions on the excess
$F(\cdot \mid X) - F(\cdot)$; pointwise convergence
$F(\cdot \mid X) \to F(\cdot)$ alone is not sufficient.

Under stronger dependence, no uniform ordering between
$\xi_{\mathrm{cut}}$ and $\xi_{AD}$ holds: the pointwise sign of
$R(u, v)$ in Lemma~\ref{lem:taylor-phi} can be either positive or
negative depending on $(u, v)$. For instance,
$\phi(0.3, 0.2) \approx 0.0282$ while
$q(0.3, 0.2) = (0.1)^2 / (2 \cdot 0.2 \cdot 0.8) = 0.0313$, giving
$R(0.3, 0.2) \approx -0.0031 < 0$. The two functionals therefore
agree to second order around independence when the tail conditions
above hold, and can diverge in either direction under stronger
alternatives.

\subsection{Connection to Dette--Siburg--Stoimenov (2013)}

The Dette--Siburg--Stoimenov (DSS) coefficient, in the
normalization used in this paper, is
\[
q_{DSS}(X, Y)
= 6 \int \E_X\!\bigl[(F(t\mid X) - F(t))^2\bigr]\, dF(t),
\]
which is a uniform-weight $L^2$ measure of CDF difference. The
factor of $6$ comes from
$\int_0^1 v(1-v)\, dv = 1/6$: under continuity of $Y$, the
denominator $\int F(t)(1 - F(t))\, dF(t) = 1/6$ is a fixed
constant, so multiplying by $6$ normalises the measure. Under
continuity of $Y$, this $q_{DSS}$ equals the population version of
Chatterjee's rank correlation $\xi_{\mathrm{Ch}}^{\mathrm{pop}}$
exactly \citep{shi2022power}, without any additional constant.
(Some earlier references, including the original DSS 2013 paper,
define the same functional without the factor $6$; the two
conventions differ by a factor and represent the same population
object.)

Comparing $q_{DSS}$ and $\xi_{AD}$ shows that $\xi_{AD}$ is the
Fisher-weighted variant of $q_{DSS}$ obtained by replacing the
uniform weight $1$ with the Bernoulli-information weight
$1/(F(t)(1 - F(t)))$. The composite Bernoulli likelihood-ratio
construction of the main paper identifies this Fisher weight as
a consequence of the likelihood, not an imposed design choice.

\subsection{The regularized coefficient}\label{sec:reg-supp}

For completeness we record properties of the regularized
coefficient introduced in Section~2.3 of the main paper. Define
\[
\xi_{\mathrm{cut}, \varepsilon}(X, Y)
\;=\; 2 \int \E_X\!\bigl[\phi_\varepsilon(F(t \mid X), F(t))\bigr]\, dF(t),
\qquad \varepsilon \ge 0,
\]
where $\phi_\varepsilon$ is defined in the Notation section. By
Lemma~\ref{lem:phi-props}, for every fixed $v \in (0, 1)$ the
map $u \mapsto \phi_\varepsilon(u, v)$ is strictly convex on
$[0, 1]$, so $\phi_\varepsilon(\cdot, v) \ge 0$ with equality iff
$u = v$; and $\phi_\varepsilon$ is bounded on $[0, 1]^2$ by
$2(1 + \varepsilon) \log((1 + \varepsilon)/\varepsilon)$.
Consequently $\xi_{\mathrm{cut}, \varepsilon}$ is well-defined
and nonnegative for every $\varepsilon > 0$ without moment
conditions on $F(\cdot \mid X)$.

For each fixed $u, v \in (0, 1)$, pointwise convergence
$\phi_\varepsilon(u, v) \to \phi(u, v)$ holds as
$\varepsilon \downarrow 0$. Passage of this limit through the
integral $\int \E_X[\phi_\varepsilon(F(t \mid X), F(t))]\, dF(t)$
to obtain $\xi_{\mathrm{cut}, \varepsilon} \to \xi_{\mathrm{cut}}$
requires an additional uniform-integrability or tail-domination
condition on the family
$\{\phi_\varepsilon(F(t \mid X), F(t)) : 0 < \varepsilon \le \varepsilon_0\}$:
the boundedness envelope
$2(1 + \varepsilon)\log((1 + \varepsilon)/\varepsilon)$ blows up
as $\varepsilon \downarrow 0$ and is not an $\varepsilon$-uniform
dominating function, so continuity and interval support alone are
not sufficient to invoke dominated convergence. Sufficient
conditions include any assumption under which the family above
is uniformly integrable with respect to $\Pr_X \otimes F$, or is
dominated by an integrable envelope independent of $\varepsilon$.
We do not require such a condition for the finite-sample
permutation procedure of Proposition~\ref{prop:perm-exact}, which
is exact for $\widehat\xi_{\mathrm{cut}, \varepsilon}$ at any
fixed $\varepsilon \ge 0$ for which the statistic is
well-defined. The upper bound corresponding to Theorem~2.1(i) is
$\xi_{\mathrm{cut}, \varepsilon} \le 2 \int_0^1 M_\varepsilon(v)\, dv$,
which is strictly less than $1$ for $\varepsilon > 0$: at
$\varepsilon = 0.01$, for example, the upper bound is
approximately $0.927$. The regularized coefficient is therefore
normalized to a scale that depends on $\varepsilon$, and for
cross-study comparability we recommend reporting the
unregularized $\xi_{\mathrm{cut}}$ ($\varepsilon = 0$) except when
numerical safeguards require otherwise.

\section{Permutation exactness}

We prove that the permutation calibration described in
Section~4.3 of the main paper is exact in finite samples, even
under the max-over-grid bandwidth rule of Section~4.2.

\begin{sproposition}[Exactness of permutation calibration]
\label{prop:perm-exact}
Assume $(X_1, Y_1), \dots, (X_n, Y_n)$ are i.i.d.\ under the null
hypothesis $H_0: X \indep Y$. Let $T_n(X, Y)$ be any statistic
that is symmetric in the sample pairs. In particular,
$T_n(X, Y) = T_{\mathrm{cut}}$ defined by
$T_{\mathrm{cut}} =
\max_{h \in \mathcal{H}(X)} \widehat\xi_{\mathrm{cut}}(h)$
qualifies, where $\mathcal{H}(X)$ is the multiplicative bandwidth
grid of Section 4.2 of the main paper.

For any permutation $\pi$ of $\{1, \dots, n\}$, define the permuted
statistic $T_n^{(\pi)} = T_n(X, Y \circ \pi)$, where $Y \circ \pi$
denotes the sample $(Y_{\pi(1)}, \dots, Y_{\pi(n)})$. Draw
$B$ permutations $\pi_1, \dots, \pi_B$ uniformly at random. The
permutation $p$-value
\begin{equation}\label{eq:pperm}
\hat p_{\mathrm{perm}}
=
\frac{1 + \#\{b : T_n^{(\pi_b)} \ge T_n\}}{B + 1}
\end{equation}
satisfies $\Pr(\hat p_{\mathrm{perm}} \le \alpha) \le \alpha$ for
every $\alpha \in (0, 1)$, exactly in finite samples.
\end{sproposition}

\begin{proof}
Under $H_0$, the pairs are i.i.d.\ with $X \indep Y$. Consider
the group $S_n$ of permutations of $\{1, \dots, n\}$ acting on
the $Y$-coordinate: for $\pi \in S_n$, define the map
$g_\pi : (X, Y) \mapsto (X, Y \circ \pi)$. Under $H_0$, the joint
distribution of the sample is invariant under this group action:
for every $\pi \in S_n$,
\[
\bigl((X_1, Y_1), \dots, (X_n, Y_n)\bigr)
\stackrel{d}{=}
\bigl((X_1, Y_{\pi(1)}), \dots, (X_n, Y_{\pi(n)})\bigr),
\]
because $X$ and $Y$ are independent and the $Y_i$ are exchangeable.

Now let $\Pi_1, \dots, \Pi_B$ be i.i.d.\ uniform on $S_n$ and
independent of the sample. By the group-invariance argument
\citep[Chapter~15]{lehmann2006testing}, the joint vector
\[
(T_n, T_n^{(\Pi_1)}, \dots, T_n^{(\Pi_B)})
\]
is exchangeable under $H_0$: any coordinate is a legitimate
$T$-value on a uniformly random permutation of the sample, and
the joint law is invariant under permutations of the vector.
(Marginal equality in distribution alone would not imply this;
the joint exchangeability requires the group action to be
compatible with the sample distribution, which is what $H_0$
provides.)

For the max-over-grid bandwidth rule specifically, the statistic
$T_n(X, Y) = \max_{h \in \mathcal{H}(X)}
\widehat\xi_{\mathrm{cut}}(h)$ satisfies the compatibility
requirement: because the grid $\mathcal{H}(X)$ depends only on
$X$ (through the Silverman-type base bandwidth $h_0$
(equation (16) of the main paper)), permuting $Y$ leaves
$\mathcal{H}(X)$ unchanged, and the recomputation
$T_n^{(\pi)}$ evaluates the same functional on the permuted
sample.

By exchangeability of $(T_n, T_n^{(\Pi_1)}, \dots, T_n^{(\Pi_B)})$
under $H_0$, the randomized rank of $T_n$ (with uniform tie-breaking)
is uniform on $\{1, \dots, B+1\}$. The permutation $p$-value
\eqref{eq:pperm} uses the non-randomized convention
$\#\{b : T_n^{(\Pi_b)} \ge T_n\}$, under which ties between $T_n$
and the permutation replicates can only make the $p$-value larger
than its randomized counterpart and hence make the test
conservative. Therefore
$\Pr(\hat p_{\mathrm{perm}} \le \alpha) \le \alpha$ for every
$\alpha \in (0, 1)$.
\end{proof}

\begin{sremark}[Scope of the exactness result]
Proposition~\ref{prop:perm-exact} establishes exactness of the
permutation $p$-value under the unconditional null
$H_0 : X \indep Y$. The result applies to
$T_n = T_{\mathrm{cut}}$, to $T_n = \widehat\xi_{AD}^*$
(with the same $X$-only bandwidth grid), and more generally to
any statistic recomputed on the permuted sample using an
$X$-only-dependent bandwidth rule.

The result does \emph{not} extend without modification to the
conditional null $H_0 : X \indep Y \mid Z$. Ordinary permutation
of $Y$ destroys the $Y$--$Z$ association and is not generally
valid for the conditional null: the marginal law of $Y$ is
preserved by permutation, but the joint law $(Y, Z)$ is not.
Testing $H_0 : X \indep Y \mid Z$ requires an appropriate
conditional-randomization scheme---for example, the conditional
randomization test of \citet{candes2018panning}, a stratified
permutation when $Z$ has finite support, or an analogous
CRT-type resampling scheme. The conditional variant
$\widehat\xi_{\mathrm{cut}}^{*\mid Z}$ introduced in Section~4.5
of the main paper is estimated by kernel smoothing over $Z$, but
its calibration should use a conditional-permutation or CRT-type
resampling scheme rather than the ordinary permutation of
Proposition~\ref{prop:perm-exact}. Full development of the
conditional testing framework is deferred to future work.
\end{sremark}

\section{Information-theoretic interpretation}\label{sec:supp-info}

This section formalizes the information-theoretic interpretation
of $\xi_{\mathrm{cut}}$ introduced in Section~2.4 of the main
paper.

\subsection{Identity as threshold-averaged mutual information}

Fix $t \in \R$ and let $B_t = \mathbf{1}(Y \le t)$. The
conditional distribution of $B_t$ given $X$ is
$\mathrm{Bernoulli}(F(t \mid X))$, and its marginal distribution
is $\mathrm{Bernoulli}(F(t))$. The mutual information
between $X$ and $B_t$ is
\begin{equation}
\label{eq:IXBt}
I(X;\, B_t)
\;=\; \E_X\!\bigl[
D_{KL}(\Pr(B_t \mid X) \,\|\, \Pr(B_t))\bigr]
\;=\; \E_X[\phi(F(t \mid X), F(t))],
\end{equation}
where the second equality uses the fact that
$\phi(u, v) = D_{KL}(\mathrm{Bernoulli}(u) \,\|\,
\mathrm{Bernoulli}(v))$. Substituting into the definition of
$\xi_{\mathrm{cut}}$ gives the identity
\begin{equation}
\label{eq:xicut-KL-supp}
\xi_{\mathrm{cut}}(X, Y)
\;=\; 2 \int I(X;\, B_t)\, dF(t),
\end{equation}
expressing $\xi_{\mathrm{cut}}$ as twice the threshold-averaged
mutual information between $X$ and $B_t$.

\subsection{Data-processing bound}

\begin{sproposition}[Data-processing bound]
\label{prop:DPI}
For every $(X, Y)$ with $Y$ continuous,
$\xi_{\mathrm{cut}}(X, Y) \le 2\, I(X; Y)$.
\end{sproposition}

\begin{proof}
Because $B_t = \mathbf{1}(Y \le t)$ is a deterministic function of
$Y$, the Markov chain $X \to Y \to B_t$ holds, and the data-
processing inequality gives $I(X; B_t) \le I(X; Y)$ for every $t$.
Integrating over $t \sim F$ and using \eqref{eq:xicut-KL-supp},
\[
\xi_{\mathrm{cut}}(X, Y)
\;=\; 2 \int I(X;\, B_t)\, dF(t)
\;\le\; 2 \int I(X; Y)\, dF(t)
\;=\; 2\, I(X; Y),
\]
since $F$ is a probability measure.
\end{proof}

For each fixed threshold $t$, the binary variable $B_t$ generally
retains only part of the information about $Y$ that is relevant
to $X$, so $I(X; B_t) \le I(X; Y)$, typically strictly. Although
the full collection $\{B_t : t \in \R\}$ determines $Y$ (via
$Y = \inf\{t : B_t = 1\}$) and hence carries the same information
about $X$ as $Y$ itself, $\xi_{\mathrm{cut}}$ averages the mutual
information of individual binary cuts rather than taking the
mutual information of the entire collection. The DPI bound is
therefore typically strict, reflecting the gap between averaging
per-threshold mutual informations and jointly using all
thresholds.

\subsection{Numerical illustration on Gaussian data}

Table~\ref{tab:info-gaussian} illustrates the bound on bivariate
standard Gaussian data with correlation $\rho$, where
$I(X; Y) = -\tfrac{1}{2} \log(1 - \rho^2)$ has a closed form.

\begin{table}[h]
\centering
\caption{Comparison of $\xi_{\mathrm{cut}}$ (computed by
numerical double integration of the population functional) with
$2 I(X; Y)$ for bivariate standard normal $(X, Y)$ with
correlation $\rho$. Both quantities are in nats.}
\label{tab:info-gaussian}
\begin{tabular}{cccc}
\toprule
$\rho$ & $\xi_{\mathrm{cut}}$ & $2 I(X; Y)$ & Ratio \\
\midrule
0.1 & 0.0048 & 0.0101 & 0.48 \\
0.3 & 0.0444 & 0.0943 & 0.47 \\
0.5 & 0.1297 & 0.2877 & 0.45 \\
0.7 & 0.2790 & 0.6733 & 0.41 \\
0.9 & 0.5570 & 1.6607 & 0.34 \\
\bottomrule
\end{tabular}
\end{table}

The bound holds strictly across the range; the ratio
$\xi_{\mathrm{cut}}/(2 I(X;Y))$ decreases from approximately
$0.48$ at weak dependence to $0.34$ at strong dependence,
reflecting the fraction of the total $X$--$Y$ mutual information
captured on average by a single random threshold.

\subsection{Practical implications}

The identity \eqref{eq:xicut-KL-supp} and
Proposition~\ref{prop:DPI} give $\xi_{\mathrm{cut}}$ a
distribution-free interpretation on the information scale:
the population coefficient is calibrated in nats (or bits, dividing by
$\log 2$), and its value serves as a distribution-free lower
bound on twice the mutual information $I(X; Y)$. This
interpretation is not available for rank-based coefficients such
as Chatterjee's $\xi$, whose values are on a scale-free
$[0, 1]$ interval without a direct information-theoretic reading.
For a population value $\xi_{\mathrm{cut}} = 0.14$, for instance,
the identity guarantees $I(X; Y) \ge 0.07$ nats. Applying this
identity to the sample estimator $T_{\mathrm{cut}}$
requires bias correction: because the max-over-bandwidth rule
introduces an upward bias of order $0.05$--$0.07$ at $n = 200$
even under independence (see Section~4.2 of the main paper), a
plug-in interpretation of $T_{\mathrm{cut}}/2$ as a
mutual-information floor overstates the true bound in absolute
terms. The permutation-calibrated $p$-value remains valid for
testing independence, since the same bias enters symmetrically
under both null and observed samples.

\section{Interpretations under semiparametric models}\label{sec:supp-semi}

Section~3 of the main paper places $\xi_{\mathrm{cut}}$ within the
Dette--Siburg--Stoimenov family of $L^2$ CDF-difference measures.
This section develops the connection to three common
semiparametric models for $F(y \mid x)$, showing that under each
$\xi_{\mathrm{cut}}$ reduces to a familiar quantity in the
weak-signal regime while remaining valid without the model.

\subsection{Location-shift model}\label{sec:supp-semi-shift}

Let $Y = \mu(X) + \varepsilon$ with $\varepsilon \indep X$
and continuous CDF $F_\varepsilon$ with bounded density
$f_\varepsilon$. For the Taylor argument below we assume
$f_\varepsilon$ is continuously differentiable with
bounded derivative (this covers Gaussian, Student-$t$, logistic
and Cauchy noise); the Laplace case, where $f_\varepsilon$ has a
kink at $0$, is handled separately by direct calculation.
Without loss of generality assume $\E[\mu(X)] = 0$ (any nonzero
mean can be absorbed into an intercept and does not affect
dependence). To formulate the local expansion rigorously we
adopt a shrinking-signal parameterization: fix a bounded,
centered function $m(\cdot)$ with $\E[m(X)] = 0$ and
$\mathrm{Var}(m(X)) < \infty$, and consider
\[
\mu_\delta(X) \;=\; \delta\, m(X),
\qquad \delta \to 0.
\]
Under this parameterization,
\[
F(y \mid X) \;=\; F_\varepsilon(y - \delta m(X)),
\qquad
F(y) \;=\; \E_X[F_\varepsilon(y - \delta m(X))],
\]
and expanding in $\delta$ (uniformly in $y$, using boundedness
of $m$ and $f_\varepsilon'$),
\[
F_\varepsilon(y - \delta m(X))
= F_\varepsilon(y) - \delta f_\varepsilon(y)\, m(X)
  + O(\delta^2),
\]
where the $O(\delta^2)$ term is bounded uniformly in $y$ by
$\tfrac{1}{2}\delta^2 \|f_\varepsilon'\|_\infty m(X)^2$. Taking
expectation gives $F(y) = F_\varepsilon(y) + O(\delta^2)$
uniformly in $y$, so the leading difference is
\[
F(y \mid X) - F(y)
\;=\; -\delta f_\varepsilon(y)\,m(X) + O(\delta^2).
\]
Squaring gives
$(F(y \mid X) - F(y))^2 = \delta^2 f_\varepsilon(y)^2 m(X)^2
+ O(\delta^3)$ uniformly in $y$, where the constant absorbed
into the $O(\delta^3)$ term depends on $\|f_\varepsilon'\|_\infty$
and $\E[m(X)^2]$. Substituting this into the second-order
Taylor expansion of the Bernoulli KL from Section~3.1 of the
main paper, and letting $F_\delta$ denote the marginal CDF of
$Y$ under the shrinking-signal model, gives to leading order
\[
\xi_{\mathrm{cut}}
\;=\; \delta^2 \E[m(X)^2] \int
\frac{f_\varepsilon(y)^2}{F_\delta(y)(1 - F_\delta(y))}\, dF_\delta(y)
+ o(\delta^2).
\]
Because $F_\delta \to F_\varepsilon$ uniformly as $\delta \to 0$,
the tail or uniform-integrability condition on the excess
$F(\cdot \mid X) - F(\cdot)$ established in Section~B (which
underlies the second-order equivalence between $\xi_{\mathrm{cut}}$
and $\xi_{AD}$) allows replacement of $F_\delta$ by
$F_\varepsilon$ in the integrand at the cost of $o(\delta^2)$.
We obtain
\begin{equation}
\label{eq:xicut-shift}
\xi_{\mathrm{cut}}
\;=\; \delta^2 C(F_\varepsilon) \cdot \mathrm{Var}\{m(X)\}
      + o(\delta^2),
\qquad
C(F_\varepsilon) := \int \frac{f_\varepsilon(t)^3}{F_\varepsilon(t)(1 - F_\varepsilon(t))}\, dt,
\end{equation}
under the integrability of the Fisher-weighted density integral
above. Only the leading first-order term of the density Taylor
expansion enters $C(F_\varepsilon)$; higher-order density terms
contribute at $O(\delta^3)$ and disappear into the $o(\delta^2)$
remainder. Equivalently, $\xi_{\mathrm{cut}} = C(F_\varepsilon)
\mathrm{Var}(\mu_\delta(X)) + o(\mathrm{Var}(\mu_\delta(X)))$.

The constant $C(F_\varepsilon)$ depends only on the noise
distribution. Two closed-form or nearly closed-form cases:
\begin{itemize}[leftmargin=1.5em, itemsep=0.25em]
\item Standard Gaussian $\varepsilon$: numerical integration gives
$C_{\mathrm{Gaussian}} \approx 0.4805$.
\item Standard Laplace $\varepsilon$ (density
$f(x) = \tfrac{1}{2} e^{-|x|}$): even though $f_\varepsilon$ has
a kink at $0$ and so falls outside the smoothness assumption above,
the population functional
$\int f_\varepsilon(t)^3 / [F_\varepsilon(t)(1 - F_\varepsilon(t))]\, dt$
is finite and can be evaluated directly. The change of variable
$u = e^{-|t|}$ yields the closed form
$C_{\mathrm{Laplace}} = \int_0^1 \tfrac{u}{2 - u}\, du = 2\log 2 - 1
\approx 0.3863$. A rigorous Taylor argument for the Laplace case
requires a separate treatment of the kink (splitting the integral
at $0$ or working with a smoothed approximation and passing to the
limit); we do not give one here.
\end{itemize}
Under the linear model $m(X) = X$ with Gaussian $X$ and
Gaussian $\varepsilon$, \eqref{eq:xicut-shift} becomes
$\xi_{\mathrm{cut}} \approx 0.481 \cdot \delta^2 \mathrm{Var}(X)$,
recovering the Pearson-$r^2$ scaling up to the noise-shape
constant $C(F_\varepsilon)$. Numerical verification at
$\delta \in \{0.1, 0.2, 0.3\}$ gives
$\xi_{\mathrm{cut}}/\delta^2 \in \{0.477, 0.467, 0.452\}$,
tending to $C_{\mathrm{Gaussian}} \approx 0.481$ as $\delta \downarrow 0$.

The reading is: under a location-shift model with Gaussian noise,
$\xi_{\mathrm{cut}}$ is a nonparametric analog of Pearson's
$r^2$, calibrated on the same $[0, 1]$ scale. For non-Gaussian
noise the same scaling holds with a different constant that
depends only on $F_\varepsilon$.

\emph{Relation to Spearman's rank correlation.} Under the
location-shift model $Y = X\beta + \varepsilon$ with
$\varepsilon \indep X$, Spearman's rank correlation provides
another rank-based measure whose squared value is of order
$\beta^2$ under a shrinking-signal regime. For scalar $X$, a
first-order expansion of $F_Y(Y) = \E_{X'}[F_\varepsilon\{\varepsilon
+ (X - X')\beta\}]$ around $\beta = 0$ gives
$F_Y(Y) = F_\varepsilon(\varepsilon) + \beta f_\varepsilon(\varepsilon)
\{X - \E(X)\} + o(\beta)$, and since $F_\varepsilon(\varepsilon)
\sim \mathrm{Unif}(0,1)$ is independent of $X$,
\begin{equation}
\label{eq:spearman-shift}
\rho_S \;=\; 12\beta\, \mathrm{Cov}\{F_X(X),\, X\}
\int f_\varepsilon(t)^2\, dt + o(\beta),
\end{equation}
using $\E\{f_\varepsilon(\varepsilon)\} = \int f_\varepsilon(t)^2\, dt$.
Consequently
$\rho_S^2 = 144\beta^2\, \mathrm{Cov}\{F_X(X), X\}^2 \bigl(\int f_\varepsilon^2\bigr)^2 + o(\beta^2)$,
which is $O(\beta^2)$ like $\xi_{\mathrm{cut}}$. The two
quantities, however, capture different aspects of the local
association. The squared Spearman correlation depends on the
squared directional alignment $\mathrm{Cov}\{F_X(X), X\}^2$,
whereas $\xi_{\mathrm{cut}}$ depends on the quadratic magnitude
$\mathrm{Var}(X)$ of the location perturbation (assuming $X$ is centered)
and a Fisher-weighted functional of the error distribution. Thus,
although both are second-order sensitive to local location
shifts, $\xi_{\mathrm{cut}}$ should not be interpreted as a
Spearman correlation or its direct transformation.

\subsection{Location-scale (heteroscedastic) model}\label{sec:supp-semi-scale}

Let $Y = \mu(X) + \sigma(X)\varepsilon$ with $\varepsilon \indep X$
and $\sigma(x) > 0$. Then
\[
F(y \mid x) \;=\; F_\varepsilon\!\left(\frac{y - \mu(x)}{\sigma(x)}\right),
\]
which changes with $x$ whenever $\mu$ or $\sigma$ does. Consider
the pure heteroscedastic case $\mu(X) \equiv 0$,
$\sigma(X) = \exp(\gamma X)$, with $\varepsilon$ symmetric around
$0$ and $\E(\varepsilon^2) < \infty$. Then $\E(Y \mid X) = 0$,
so $\mathrm{Cov}(X, Y) = 0$ and Pearson's correlation vanishes
exactly. Rank-based measures that do not use the conditional CDF
(e.g.\ Spearman, Kendall) are also zero.

\emph{Why Pearson misses this dependence.} Pearson's correlation
is a functional of the joint distribution that depends only on
the first conditional moment: from
$\mathrm{Cov}(X, Y) = \E[X \cdot \E(Y \mid X)] - \E(X)\E(Y)$,
whenever $\E(Y \mid X)$ is constant in $X$, Pearson's coefficient
is zero regardless of any higher-order dependence structure.
Under the pure heteroscedastic model, $\E(Y \mid X) = \sigma(X)
\E(\varepsilon) = 0$ by symmetry, so Pearson is blind to the
scale variation by construction. Spearman and Kendall inherit
the same limitation: after rank transformation, symmetry of the
conditional distribution around its median still yields zero
directional alignment.

\emph{Why $\xi_{\mathrm{cut}}$ detects it.} $\xi_{\mathrm{cut}}$
is a functional of the full conditional CDF, not of any single
conditional moment. The threshold-wise integrand
$\phi(F(t \mid X), F(t))$ is strictly positive whenever
$F(t \mid X) \ne F(t)$, and under $\sigma(X) = \exp(\gamma X)$
the map $t \mapsto F(t \mid X = x)$ is a rescaling of
$t \mapsto F_\varepsilon(t)$ by $\exp(-\gamma x)$: heavier tails
for large $\sigma(X)$, sharper concentration for small $\sigma(X)$.
Consequently $F(t \mid X) - F(t) \ne 0$ at every $t$ except the
common median, and $\E[\phi(F(t \mid X), F(t))] > 0$ at every
non-median $t$. Integrating over $t \sim F$ gives
$\xi_{\mathrm{cut}} > 0$. The same argument shows
$\xi_{\mathrm{Ch}}^{\mathrm{pop}} = q_{DSS} > 0$: both members of
the DSS family detect the scale-driven CDF perturbation that
Pearson's mean-based construction cannot see.

Both $\xi_{\mathrm{Ch}}$ and $\xi_{\mathrm{cut}}$ remain
positive, however, because the conditional CDF $F(y \mid x)$
genuinely depends on $x$. From Section~3.2, both population
functionals are of the form
$\int \E[(F(t \mid X) - F(t))^2 W(t)]\, dF(t)$; the distinction
is the weight $W(t)$: $q_{DSS} = \xi_{\mathrm{Ch}}^{\mathrm{pop}}$
uses a uniform weight (up to the multiplier $6$), whereas the
second-order form of $\xi_{\mathrm{cut}}$ uses the Fisher weight
$1/[F(t)(1 - F(t))]$. The Fisher weight up-weights threshold
locations where $F(t) \approx 0$ or $F(t) \approx 1$, i.e.\ the
tails of the marginal distribution of $Y$, which can emphasize
tail departures of $F(t \mid X)$ from $F(t)$ produced by scale
changes. This is a plausible explanation for the empirical
finding in the S6 heteroscedastic scenario of the main paper's
simulation (Section~5), where $\xi_{\mathrm{cut}}$ achieves
3-4$\times$ the power of the empirical $\widehat\xi_{\mathrm{Ch}}$
despite both population coefficients being positive.

Monte Carlo evaluation of the population functional (Gaussian
$\varepsilon$, $X \sim N(0,1)$) confirms sensitivity: 
$\xi_{\mathrm{cut}}(\gamma) \approx \{0.011, 0.048, 0.108\}$
at $\gamma = \{0.2, 0.5, 1.0\}$, while Pearson's $r$ remains
within Monte Carlo noise of zero.

\subsection{Cox proportional hazards model}\label{sec:supp-semi-cox}

Let $Y$ have conditional hazard
$\lambda(t \mid x) = \lambda_0(t) e^{\beta x}$, so
\[
F(t \mid x) \;=\; 1 - S_0(t)^{\exp(\beta x)},
\qquad
S_0(t) = \exp\!\left(-\int_0^t \lambda_0(s)\, ds\right).
\]
Without loss of generality assume $\E(X) = 0$; otherwise replace
$x$ by $x - \E(X)$, absorbing the shift into the baseline hazard.
Expand $S_0(t)^{\exp(\beta x)} = \exp(e^{\beta x}\log S_0(t))$
around $\beta = 0$:
\[
S_0(t)^{\exp(\beta x)}
= S_0(t)\bigl[1 + \beta x \log S_0(t) + O(\beta^2)\bigr].
\]
Under $\E(X) = 0$, taking expectation gives
$F(t) = 1 - S_0(t) + O(\beta^2) = F_0(t) + O(\beta^2)$, so
\[
F(t \mid X) - F(t)
\;=\; -\beta X\, S_0(t)\log S_0(t) + O(\beta^2).
\]
Substituting into the second-order Taylor expansion of
$\xi_{\mathrm{cut}}$ yields
\[
\xi_{\mathrm{cut}}
\;=\; \beta^2 \mathrm{Var}(X) \cdot
\int \frac{[S_0(t) \log S_0(t)]^2}{F_0(t)(1 - F_0(t))}\, dF_0(t)
+ o(\beta^2).
\]
Now the key observation: change variables $v = F_0(t)$ so
$S_0(t) = 1 - v$ and $dF_0(t) = dv$:
\begin{align*}
\int \frac{[S_0(t) \log S_0(t)]^2}{F_0(t)(1-F_0(t))} dF_0(t)
&= \int_0^1 \frac{(1-v)^2 \log^2(1-v)}{v(1-v)}\, dv \\
&= \int_0^1 \frac{(1-v)\log^2(1-v)}{v}\, dv
\;=\; 2\{\zeta(3) - 1\},
\end{align*}
where the last equality follows from the substitution $u = 1-v$
and the series expansion $u/(1-u) = \sum_{k \ge 1} u^k$ combined
with $\int_0^1 u^k \log^2 u\, du = 2/(k+1)^3$. The baseline
distribution $F_0$ has cancelled entirely. We therefore obtain
the universal expansion
\begin{equation}
\label{eq:xicut-cox}
\boxed{\;
\xi_{\mathrm{cut}}
\;=\; 2\{\zeta(3) - 1\}\, \beta^2 \mathrm{Var}(X) + o(\beta^2),
\qquad 2\{\zeta(3) - 1\} \approx 0.4041,
\;}
\end{equation}
valid for any continuous baseline hazard and any centered
covariate distribution with sufficient local exponential-moment
regularity, for example $\E[e^{\eta |X|}] < \infty$ for some
$\eta > 0$ (which controls the remainder in the expansion of
$e^{\beta X}$ uniformly in a neighborhood of $\beta = 0$).
The condition is trivially satisfied for bounded $X$ and for
Gaussian, sub-Gaussian, or sub-exponential covariates; it
excludes heavy-tailed $X$ such as Cauchy. Monte Carlo
verification with Weibull baselines of shape $k \in \{0.5, 1, 2\}$
and $\beta \in \{0.05, 0.1\}$ gives ratios
$\xi_{\mathrm{cut}}/\beta^2$ clustering around $0.40$
across all baselines, consistent with \eqref{eq:xicut-cox}.

The universality of the constant $2\{\zeta(3) - 1\}$ is a genuine
strength: under the Cox model, the local behavior of
$\xi_{\mathrm{cut}}$ depends only on $\beta$ and $\mathrm{Var}(X)$,
not on any nuisance shape of $\lambda_0$.

\emph{Structural connection to log-rank aggregation.} For each
threshold $t$, the Bernoulli KL contribution
$\E[\phi(F(t \mid X), F(t))]$ is the population log-likelihood
ratio comparing the conditional probability
$\Pr(Y \le t \mid X)$ to the marginal probability $\Pr(Y \le t)$,
averaged over $X$. This is analogous, but not literally identical,
to the stratified $2 \times 2$ contributions aggregated by the
log-rank statistic in survival analysis (where the strata are
formed by discrete event times and the tables classify observed
event indicators). For continuous $X$, the contribution at each
$t$ is a mixture over conditional-probability comparisons rather
than a single $2 \times 2$ table. Aggregating these contributions
across $t \sim F$ under the Cox model gives $\xi_{\mathrm{cut}}$
its interpretation as an integrated threshold-level departure
from independence.

\subsection{Summary}

Under three common semiparametric families:
\begin{itemize}[leftmargin=1.5em, itemsep=0.25em]
\item \emph{Location-shift}: $\xi_{\mathrm{cut}}$ is a
nonparametric analog of Pearson's $R^2$, with the proportionality
constant depending on the noise distribution shape.
\item \emph{Location-scale}: pure heteroscedasticity gives
Pearson $r = 0$ exactly, while $\xi_{\mathrm{cut}}$ and
Chatterjee's $\xi_{\mathrm{Ch}}$ both detect it at the population
level; the Fisher weighting in $\xi_{\mathrm{cut}}$ can emphasize
tail thresholds affected by scale changes, providing a plausible
explanation for the finite-sample power gap in favor of
$\xi_{\mathrm{cut}}$.
\item \emph{Cox proportional hazards}: $\xi_{\mathrm{cut}}$ admits
the universal expansion
$\xi_{\mathrm{cut}} = 2\{\zeta(3) - 1\} \beta^2 \mathrm{Var}(X)
+ o(\beta^2)$, with a constant $\approx 0.404$ that does not
depend on the baseline hazard.
\end{itemize}
All three statements are local (small dependence) approximations;
the exact identities do not simplify under strong dependence.
The identities are useful for interpreting $\xi_{\mathrm{cut}}$
values on real data where a semiparametric model is a reasonable
working assumption, without requiring the model for validity.

\section{Additional simulation results}

\subsection{Empirical size and power at additional sample sizes}

Table~\ref{tab:sim-n100} reports the analogue of the main paper's
Table~1 at $n = 100$; Table~\ref{tab:sim-n500} reports the analogue
at $n = 500$. Both use $M = 1000$ Monte Carlo replications with
$B = 999$ permutations at $\lambda = 0.5$.

\begin{table}[h]
\centering
\small
\caption{Empirical power at nominal level $\alpha = 0.05$ at
$n = 100$, $\lambda = 0.5$. Compare with Table~1 of the main
paper ($n = 200$).}
\label{tab:sim-n100}
\begin{tabular}{lccccc}
\toprule
Scenario
  & $\xi_{\text{Ch}}$
  & $\nu_{\text{AR}}$
  & $\xi_{\text{Ch,boost}}^{M=50}$
  & $\widehat\xi_{AD}^*$
  & $T_{\mathrm{cut}}$ \\
\midrule
S1 Null                    & 0.052 & 0.056 & 0.032 & 0.056 & 0.055 \\
S2 Linear                  & 0.086 & 0.090 & 0.548 & 0.203 & 0.202 \\
S3 Threshold               & 0.856 & 0.836 & 0.999 & 0.994 & 0.995 \\
S4 W-shape                 & 0.088 & 0.084 & 0.033 & 0.166 & 0.171 \\
S5 Oscillating             & 0.997 & 0.996 & 0.000 & 1.000 & 0.999 \\
S6 Heteroscedastic         & 0.083 & 0.106 & 0.107 & 0.202 & 0.170 \\
\bottomrule
\end{tabular}
\end{table}

\begin{table}[h]
\centering
\small
\caption{Empirical power at nominal level $\alpha = 0.05$ at
$n = 500$, $\lambda = 0.5$. Compare with Table~1 of the main
paper ($n = 200$).}
\label{tab:sim-n500}
\begin{tabular}{lccccc}
\toprule
Scenario
  & $\xi_{\text{Ch}}$
  & $\nu_{\text{AR}}$
  & $\xi_{\text{Ch,boost}}^{M=50}$
  & $\widehat\xi_{AD}^*$
  & $T_{\mathrm{cut}}$ \\
\midrule
S1 Null                    & 0.042 & 0.048 & 0.046 & 0.048 & 0.041 \\
S2 Linear                  & 0.176 & 0.184 & 0.936 & 0.805 & 0.805 \\
S3 Threshold               & 1.000 & 1.000 & 1.000 & 1.000 & 1.000 \\
S4 W-shape                 & 0.175 & 0.171 & 0.285 & 0.790 & 0.789 \\
S5 Oscillating             & 1.000 & 1.000 & 0.504 & 1.000 & 1.000 \\
S6 Heteroscedastic         & 0.134 & 0.219 & 0.610 & 0.953 & 0.966 \\
\bottomrule
\end{tabular}
\end{table}

The qualitative ranking observed at $n = 200$ (Table~1 of the
main paper) is preserved across sample sizes:
$T_{\mathrm{cut}}$ substantially outperforms
$\xi_{\mathrm{Ch}}$ and $\nu_{\mathrm{AR}}$ on the non-monotone
scenarios (S4 and S6), matches or exceeds them on the monotone
alternative (S2), and maintains correct empirical size under the
null (S1). The boosted Chatterjee variant
$\xi_{\text{Ch,boost}}^{M=50}$ outperforms all methods on the
linear alternative (S2) at every sample size but exhibits a marked
small-subsample failure pattern on the oscillating alternative
(S5): at
$n = 100$ its power is $0.000$, at $n = 200$ still $0.000$, and
at $n = 500$ it recovers only partially to $0.504$, while
$T_{\mathrm{cut}}$ retains full power $\approx 1.000$
at every sample size. This reflects the boosted statistic's
specialization for smooth monotone dependence and its
degradation on periodic signals until the sample size becomes
large enough for the boosting kernel to resolve the oscillations.

\subsection{Empirical size and power at additional signal-to-noise levels}

Table~\ref{tab:sim-lambda025} reports the analogue of the main
paper's Table~1 at the stronger signal-to-noise setting
$\lambda = 0.25$; Table~\ref{tab:sim-lambda075} reports the
analogue at the weaker setting $\lambda = 0.75$. Both use $n = 200$,
$M = 1000$ Monte Carlo replications, and $B = 999$ permutations.

\begin{table}[h]
\centering
\small
\caption{Empirical power at nominal level $\alpha = 0.05$ at
$n = 200$, $\lambda = 0.25$ (strong signal-to-noise). Compare
with Table~1 of the main paper ($\lambda = 0.5$). Under the
S6 design $Y = (1 + 2|X|)\lambda \varepsilon$, $\lambda$
enters purely as an overall scale factor; because the
coefficients considered are scale-invariant, the S6 row is
$\lambda$-invariant in expectation. The finite-sample values
shown here differ from the corresponding entry in Table~1 by
at most $0.001$, within Monte Carlo error at $M = 1000$
replications.}
\label{tab:sim-lambda025}
\begin{tabular}{lccccc}
\toprule
Scenario
  & $\xi_{\text{Ch}}$
  & $\nu_{\text{AR}}$
  & $\xi_{\text{Ch,boost}}^{M=50}$
  & $\widehat\xi_{AD}^*$
  & $T_{\mathrm{cut}}$ \\
\midrule
S1 Null                    & 0.041 & 0.051 & 0.043 & 0.049 & 0.051 \\
S2 Linear                  & 0.490 & 0.497 & 1.000 & 0.961 & 0.960 \\
S3 Threshold               & 1.000 & 1.000 & 1.000 & 1.000 & 1.000 \\
S4 W-shape                 & 0.475 & 0.469 & 0.053 & 0.955 & 0.956 \\
S5 Oscillating             & 1.000 & 1.000 & 0.000 & 1.000 & 1.000 \\
S6 Heteroscedastic         & 0.089 & 0.136 & 0.258 & 0.445 & 0.419 \\
\bottomrule
\end{tabular}
\end{table}

\begin{table}[h]
\centering
\small
\caption{Empirical power at nominal level $\alpha = 0.05$ at
$n = 200$, $\lambda = 0.75$ (weak signal-to-noise). Compare
with Table~1 of the main paper ($\lambda = 0.5$). The S6 row is
$\lambda$-invariant because $\lambda$ enters as a pure scale
factor and the coefficients are scale-invariant.}
\label{tab:sim-lambda075}
\begin{tabular}{lccccc}
\toprule
Scenario
  & $\xi_{\text{Ch}}$
  & $\nu_{\text{AR}}$
  & $\xi_{\text{Ch,boost}}^{M=50}$
  & $\widehat\xi_{AD}^*$
  & $T_{\mathrm{cut}}$ \\
\midrule
S1 Null                    & 0.041 & 0.051 & 0.043 & 0.049 & 0.051 \\
S2 Linear                  & 0.061 & 0.071 & 0.453 & 0.186 & 0.191 \\
S3 Threshold               & 0.655 & 0.666 & 0.998 & 0.993 & 0.994 \\
S4 W-shape                 & 0.064 & 0.068 & 0.044 & 0.163 & 0.164 \\
S5 Oscillating             & 0.968 & 0.958 & 0.000 & 1.000 & 1.000 \\
S6 Heteroscedastic         & 0.089 & 0.136 & 0.258 & 0.445 & 0.419 \\
\bottomrule
\end{tabular}
\end{table}

The qualitative ordering established at $\lambda = 0.5$ persists
at both stronger and weaker signal-to-noise. At $\lambda = 0.25$
the gap between $T_{\mathrm{cut}}$ and the non-parametric
competitors is particularly striking on both the linear and
non-monotone alternatives: on S2 Linear,
$T_{\mathrm{cut}} = 0.960$ against $\nu_{\text{AR}} = 0.497$
and $\xi_{\text{Ch}} = 0.490$ (roughly $2\times$ power); on S4
W-shape, $T_{\mathrm{cut}} = 0.956$ against
$\nu_{\text{AR}} = 0.469$ and $\xi_{\text{Ch}} = 0.475$. At
$\lambda = 0.75$ all methods lose absolute power but the ordering
is unchanged. The boosted variant's small-subsample failure pattern on S5
persists across $\lambda$: its power is $0.000$ for
$\lambda \in \{0.25, 0.5, 0.75\}$ at $n = 200$, independent of
signal-to-noise.

\subsection{Sensitivity to the regularization parameter
\texorpdfstring{$\varepsilon$}{eps}}

Table~\ref{tab:eps-sensitivity} compares empirical power of
$T_{\mathrm{cut}}$ at the unregularized value
$\varepsilon = 0$ against the safeguard value $\varepsilon = 0.05$
at $n = 200$, $\lambda = 0.5$. The operating characteristics
change only marginally between the two values, with the largest
difference on the heteroscedastic scenario (S6: $0.420$ versus
$0.358$). In the moderate signal-to-noise settings considered
here the two produce comparable behavior, and
$\varepsilon = 0$ is the default in the main paper
(Section 2.3).

\begin{table}[h]
\centering
\small
\caption{Sensitivity of $T_{\mathrm{cut}}$ to
regularization parameter $\varepsilon$, $n = 200$, $\lambda = 0.5$,
$M = 1000$ Monte Carlo replications, $B = 999$ permutations.}
\label{tab:eps-sensitivity}
\begin{tabular}{lcc}
\toprule
Scenario & $\varepsilon = 0$ & $\varepsilon = 0.05$ \\
\midrule
S1 Null                    & 0.051 & 0.050 \\
S2 Linear                  & 0.386 & 0.383 \\
S3 Threshold               & 1.000 & 1.000 \\
S4 W-shape                 & 0.321 & 0.320 \\
S5 Oscillating             & 1.000 & 1.000 \\
S6 Heteroscedastic         & 0.420 & 0.358 \\
\bottomrule
\end{tabular}
\end{table}

\bibliographystyle{plainnat}

\section{Aging plasma proteome application: computational details}\label{sec:supp-aging}

This appendix documents the computational details of the
Section~6 application.

\subsection{Data source and access}\label{sec:supp-aging-data}

The application uses the \texttt{agingplasmaproteome} example
dataset distributed with the \texttt{DEswan} R package
\citep{shen2020deswan}. The data can be obtained without a
package installation directly from the package's GitHub
repository:

\begin{center}
\small
\verb|https://github.com/lehallib/DEswan/raw/master/data/agingplasmaproteome.rda|
\end{center}

The full DEswan example dataset contains $171$ subjects across
four cohorts (Seattle: $n = 70$, PRIN06: $n = 51$, PRIN09:
$n = 24$, GEHA: $n = 26$), with $1{,}305$ plasma proteins
measured on each subject. Section~6 uses the Seattle cohort
subset ($n = 70$, ages $21$--$88$) for the reasons discussed in Section~6.1 of the main paper. In the
Seattle subset, age takes $47$ distinct values across the
$70$ subjects. Chatterjee's $\xi_{\text{Ch}}$ ties in $X$ are
broken uniformly at random (fresh permutation of tied ranks
per statistic call), following the \texttt{XICOR} R package
convention.

\subsection{Single-file reproduction script}\label{sec:supp-aging-code}

A self-contained R script that downloads the data and runs the
Seattle-only analysis reported in Section~6 is provided in
\texttt{aging\_analyses.R}. The script uses parallel processing
via \texttt{parallel::mclapply} and completes in approximately
$4$ minutes on a $30$-core node. It produces two output files:
\texttt{aging\_seattle\_results.csv} (per-protein $p$-values
from all four methods) and \texttt{aging\_analyses\_output.txt}
(a human-readable summary). All four methods---Pearson,
Spearman, Chatterjee's $\xi_{\text{Ch}}$, and $T_{\mathrm{cut}}$
---use exact permutation calibration with $B = 999$ replicates,
matching the unadjusted and Benjamini--Hochberg counts reported
in Table~2 of the main paper.

\end{document}